\documentclass[runningheads,envcountsame]{llncs}

\usepackage[T1]{fontenc}
\usepackage[utf8]{inputenc}

\usepackage{amsmath}
\usepackage{amssymb}
\usepackage{bbold} 
\usepackage{graphicx}
\usepackage{color}
\usepackage[pdftex,bookmarks,hidelinks,breaklinks]{hyperref} 
\usepackage{listings}
\usepackage{mathabx}
\usepackage{mathpartir} 
\usepackage{stmaryrd}
\usepackage{mathtools} 
\usepackage{upquote}
\usepackage[dvipsnames]{xcolor}
\usepackage[all]{xy}
\usepackage[capitalise,nameinlink]{cleveref} 

\spnewtheorem{definition}{Definition}{\bfseries}{\rmfamily}

\begin{document}



\newcommand{\defeq}{\mathrel{\overset{\text{\tiny def}}{=}}} 

\newcommand{\pl}[1]{\textsc{#1}} 

\newcommand{\lambdacoop}{\lambda_{\mathsf{coop}}} 

\newcommand{\bnfis}{\mathrel{\;{:}{:}{=}\ }}
\newcommand{\bnfor}{\mathrel{\;\big|\ \ }}



\newcommand{\sig}{\Sigma} 

\newcommand{\Tree}[2]{\mathrm{Tree}_{#1}\left(#2\right)} 
\newcommand{\retTree}[1]{\mathsf{return}\,#1} 

\newcommand{\eq}{\mathrm{Eq}} 

\newcommand{\Th}{\mathcal{T}} 
\newcommand{\Thsig}{\sig_\Th} 
\newcommand{\Theq}{\eq_\Th} 

\newcommand{\ThUU}[2]{\mathcal{U}_{#1, #2}} 
\newcommand{\UU}[2]{\mathsf{U}_{#1, #2}} 

\newcommand{\ThKK}[4]{\mathcal{K}_{#2, #3, #4, #1}} 
\newcommand{\KK}[4]{\mathsf{K}_{#2, #3, #4, #1}} 

\newcommand{\UUskel}{\skel{\mathsf{U}}} 
\newcommand{\KKskel}[1]{\skel{\mathsf{K}}_{#1}} 

\newcommand{\Runner}[4]{\mathsf{Runner}_{#1, #2, #3}\, #4} 
\newcommand{\RunnerSkel}[1]{\skel{\mathsf{Runner}}\,#1} 

\newcommand{\FreeAlg}[2]{\mathrm{Free}_{#1}\left(#2\right)} 
\newcommand{\lift}[1]{#1^\dagger} 
\newcommand{\freelift}[1]{#1^\ddagger} 

\newcommand{\M}{\mathcal{M}} 
\newcommand{\Mcarrier}{\vert \mathcal{M} \vert} 

\newcommand{\T}{T} 
\newcommand{\St}[1]{\mathsf{St}_{#1}} 
\newcommand{\Exc}[1]{\mathsf{Exc}_{#1}} 

\newcommand{\R}{\mathcal{R}} 
\newcommand{\Rcarrier}{\vert\R\vert} 
\newcommand{\iRcarrier}[1]{\vert\R_{#1}\vert} 
\newcommand{\runh}{\mathsf{r}} 


\newcommand{\expto}{\Rightarrow} 
\newcommand{\lam}[1]{\lambda #1 \,.\,} 

\newcommand{\One}{\mathbb{1}} 
\newcommand{\Zero}{\mathbb{0}} 

\newcommand{\op}{\mathsf{op}} 
\newcommand{\coop}{\mathsf{\overline{op}}} 
\newcommand{\opto}{\leadsto} 
\newcommand{\tysigop}[3]{#1 \opto #2 \mathbin{!} #3} 

\newcommand{\sigraise}{\mathsf{raise}} 
\newcommand{\sigkill}{\mathsf{kill}} 
\newcommand{\siggetenv}{\mathsf{getenv}} 
\newcommand{\sigsetenv}{\mathsf{setenv}} 


\newcommand{\at}{\mathbin{@}} 
\newcommand{\atkernel}{\mathbin{\!\lightning\!}} 
\newcommand{\atuser}{\mathbin{!}} 

\newcommand{\tybase}{\mathsf{b}} 
\newcommand{\tyunit}{\mathsf{unit}} 
\newcommand{\tyempty}{\mathsf{empty}} 
\newcommand{\typrod}[2]{#1 \times #2} 
\newcommand{\tysum}[2]{#1 + #2} 
\newcommand{\tyfun}[2]{#1 \to #2} 
\newcommand{\tyfunK}[2]{#1 \to #2} 
\newcommand{\tyrunner}[4]{#1 \Rightarrow (#2, #3, #4)} 


\newcommand{\Ueff}{\mathcal{U}} 
\newcommand{\Veff}{\mathcal{V}} 
\newcommand{\Keff}{\mathcal{K}} 
\newcommand{\Leff}{\mathcal{L}} 

\newcommand{\uty}[1]{\overline{#1}} 
\newcommand{\kty}[1]{\overline{\overline{#1}}} 

\newcommand{\tyuser}[2]{#1 \atuser #2} 
\newcommand{\tykernel}[2]{#1 \atkernel #2} 


\newcommand{\skel}[1]{{#1}^{\mathsf{s}}} 
\newcommand{\skeleton}[1]{{#1}^{\mathsf{s}}} 

\newcommand{\tyrunnerskel}[1]{\mathsf{runner}\,C} 
\newcommand{\tyfunskel}[2]{#1 \to #2} 
\newcommand{\tyfunKskel}[2]{#1 \to #2} 

\newcommand{\tyuserskel}[1]{#1{!}} 
\newcommand{\tykernelskel}[2]{#1 \atkernel #2} 


\newcommand{\Ops}{\mathcal{O}} 
\newcommand{\Sigs}{\mathcal{S}} 
\newcommand{\Excs}{\mathcal{E}} 

\newcommand{\tm}[1]{\mathsf{#1}} 
\newcommand{\tmkw}[1]{\tm{\color{keywordColor}#1}} 

\newcommand{\tmconst}[1]{\tm{#1}}
\newcommand{\tmunit}{()} 
\newcommand{\tmpair}[2]{( #1 , #2 )} 
\newcommand{\tminl}[2][]{\tmkw{inl}_{#1}\,#2} 
\newcommand{\tminr}[2][]{\tmkw{inr}_{#1}\,#2} 
\newcommand{\tmfun}[2]{{\mathop{\tmkw{fun}}}\; (#1) \mapsto #2} 
\newcommand{\tmfunK}[2]{{\mathop{\tmkw{funK}}}\; (#1) \mapsto #2} 
\newcommand{\tmapp}[2]{#1\,#2} 

\newcommand{\tmrunner}[2]{\{#1\}_{#2}} 

\newcommand{\tmreturn}[2][]{\tmkw{return}_{#1}\, #2} 
\newcommand{\tmlet}[4][]{\tmkw{let}_{#1}\; #2 = #3 \;\tmkw{in}\; #4} 
\newcommand{\tmtry}[2]{\tmkw{try}\; #1 \; \tmkw{with}\; #2}
\newcommand{\tmkernel}[3]{\tmkw{kernel} \; #1 \at #2 \;\tmkw{finally} \; #3}
\newcommand{\tmuser}[2]{\tmkw{user} \; #1 \; \tmkw{with} \; #2}
\newcommand{\tmrun}[4]{\tmkw{using} \; #1 \at #2 \; \tmkw{run} \; #3 \; \tmkw{finally} \; #4}

\newcommand{\tmop}[5]{\tm{#1}_{#2}(#3, #4, #5)} 
\newcommand{\tmgeneff}[2]{\tm{#1}\; #2} 
\newcommand{\tmcont}[2]{(#1 \,.\, #2)} 
\newcommand{\tmexccont}[3]{(#1_{#2})_{#2 \in #3}} 

\newcommand{\tmexcept}[3]{\tmreturn{#1} \mapsto #2, #3} 
\newcommand{\tmfinally}[4]{\{\tmreturn{#1} \mapsto #2, #3, #4\}} 

\newcommand{\tmmatch}[3][]{\tmkw{match}\;#2\;\tmkw{with}\;\{#3\}_{#1}} 

\newcommand{\tmraise}[2][]{\tmkw{raise}_{#1}\,#2} 
\newcommand{\tmkill}[2][]{\tmkw{kill}_{#1}\,#2} 

\newcommand{\tmgetenv}[2][]{\tmkw{getenv}_{#1}#2} 
\newcommand{\tmsetenv}[2]{\tmkw{setenv}(#1, #2)} 


\newcommand{\types}{\vdash} 
\newcommand{\typesskel}{\vdash^{\!\mathsf{s}}} 
\newcommand{\of}{\mathinner{:}} 

\newcommand{\vj}[3]{#1 \vdash #2 : #3} 
\newcommand{\cj}[5]{#1 \vdash^{\kern -1.25ex #2} #3 : #4 \mathrel{!} #5} 
\newcommand{\kj}[7]{#1 \vdash^{\kern -1.25ex #2} #3 : #4 \mathrel{!} #5 \mathrel{\lightning} #6 \at #7} 

\newcommand{\subty}[2]{#1 \sqsubseteq #2} 
\newcommand{\sub}{\sqsubseteq} 

\newcommand{\mkrule}[3]{\frac{#1}{#3}{\textsc{#2}}} 

\newcommand{\coopinfer}[3]{\inferrule*[Lab={\color{rulenameColor}#1}]{#2}{#3}}


\newcommand{\sem}[1]{[\![\![#1]\!]\!]} 
\newcommand{\skelsem}[1]{[\![#1]\!]} 

\newcommand{\subexpto}{\Rrightarrow} 

\newcommand{\Set}{\mathrm{Set}} 
\newcommand{\Cpo}{\mathrm{Cpo}} 
\newcommand{\Sub}{\mathrm{Sub}} 

\newcommand{\cond}[3]{\mathsf{if}\;#1\;\mathsf{then}\;#2\;\mathsf{else}\;#3} 


\title{Runners in action}

\author{Danel Ahman \and Andrej Bauer}
\institute{%
  Faculty of Mathematics and Physics\\
  University of Ljubljana, Slovenia}

\maketitle

\begin{abstract}
Runners of algebraic effects, also known as comodels, provide a mathematical
model of resource management. We show that they also give rise to a programming
concept that models top-level external resources, as well as allows programmers
to modularly define their own intermediate ``virtual machines''.
We capture the core ideas of programming with runners in an equational calculus $\lambdacoop$, which we equip  
with a sound and coherent denotational semantics that guarantees the linear use of
resources and execution of finalisation code. We accompany $\lambdacoop$ with
examples of runners in action, provide a prototype language implementation in \pl{OCaml},
as well as a \pl{Haskell} library based on $\lambdacoop$.
\end{abstract}



\definecolor{codegreen}{rgb}{0,0.6,0}
\definecolor{codegray}{rgb}{0.5,0.5,0.5}
\definecolor{codepurple}{rgb}{0.58,0,0.82}
\definecolor{backcolour}{rgb}{0.95,0.95,0.92}

\colorlet{keywordColor}{NavyBlue} 
\colorlet{rulenameColor}{Gray} 

\def\lstlanguagefiles{coop.tex}
\lstset{language=coop,upquote=true}
\let\ls\lstinline



\crefformat{section}{\S#2#1#3}
\Crefformat{section}{\S#2#1#3}

\crefformat{subsection}{\S#2#1#3}
\Crefformat{subsection}{\S#2#1#3}

\crefformat{subsubsection}{\S#2#1#3}
\Crefformat{subsubsection}{\S#2#1#3}

\crefformat{theorem}{Thm.~#2#1#3}
\Crefformat{theorem}{Thm.~#2#1#3}

\crefformat{proposition}{Prop.~#2#1#3}
\Crefformat{proposition}{Prop.~#2#1#3}

\crefformat{figure}{Fig.~#2#1#3}
\Crefformat{figure}{Fig.~#2#1#3}


\section{Introduction}
\label{sect:introduction}

Computational effects, such as exceptions, input-output, state, nondeterminism, and
randomness, are an important component of general-purpose programming languages,
whether they adopt functional, imperative, object-oriented, or other
programming paradigms. Even pure languages exhibit
computational effects at the top level, so to speak, by interacting with their
external environment.

In modern languages, computational effects are often
structured using~\emph{monads}~\cite{Moggi:ComputationalLambdaCalculus,Moggi:NotionsofComputationandMonads,Wadler:Monads},
or \emph{algebraic effects and handlers}~\cite{Kammar:Handlers,Plotkin:NotionsOfComputation,Plotkin:HandlingEffects}.
These mechanisms excel at implementation of computational effects within
the language itself. For instance, the familiar implementation of mutable state in terms of state-passing functions requires no native state, and can be implemented either as a monad or using handlers.
One is naturally drawn to using these techniques also for dealing with actual effects, such as manipulation of native memory and access to  hardware. These are represented inside the language as algebraic operations (as in \pl{Eff}~\cite{Bauer:AlgebraicEffects}) or a monad (in the style of \pl{Haskell}'s~\textsf{IO}), but treated specially by the language's top-level runtime, which invokes corresponding operating system functionality.
While this approach works in practice, it has some unfortunate downsides too, namely \emph{lack of modularity and
linearity}, and \emph{excessive generality}.

Lack of modularity is caused by having the external resources hard-coded into the top-level runtime.
As a result, changing which resources are available and how they are implemented requires modifications of the language implementation. Additional complications arise when a language supports several operating systems and hardware platforms, each providing their own, different feature set. One wishes that the ingenuity of the language implementors were better supported by a more flexible methodology with a sound theoretical footing.

Excessive generality is not as easily discerned, because generality of programming concepts makes a language expressive and useful, 
such as general algebraic effects and handlers enabling one to implement timeouts, rollbacks, stream redirection~\cite{Plotkin:HandlingEffects}, async \& await~\cite{Leijen:AsyncAwait}, and concurrency~\cite{Dolan:MulticoreOCaml}. However, the flip side of such expressive freedom is the lack of any guarantees about how external resources will actually be used.
For instance, consider a simple piece of code, written in \pl{Eff}-like syntax, which first opens a file, then writes to it, and finally closes it:
\begin{lstlisting}
let fh = open "hello.txt" in write (fh, "Hello, world."); close fh
\end{lstlisting}

What this program actually does depends on how the operations $\tm{open}$, $\tm{write}$, and $\tm{close}$ are handled. For all we know, an enveloping handler may intercept the $\tm{write}$ operation and discard its continuation, so that $\tm{close}$ never happens and the file is not properly closed. 
Telling the programmer not to shoot themselves in the foot by avoiding such handlers is not helpful, because the handler may encounter an external reason for not being able to continue, say a full disk.

Even worse, external resources may be misused accidentally when we combine two handlers, each of which works as intended on its own. 
For example, if we combine the above code with a non-deterministic $\tm{choose}$ operation, as in
\begin{lstlisting}
let fh = open "greeting.txt" in
let b = choose () in
if b then write (fh, "hello") else write (fh, "good bye") ; close fh
\end{lstlisting}
and handle it with the standard non-determinism handler
\begin{lstlisting}
handler { return x -> [x], choose () k -> return (append (k true) (k false)) }
\end{lstlisting}
The resulting program attempts to close the file twice, as well as write to it twice, because the continuation $\tm{k}$ is invoked twice when handling $\tm{choose}$.
Of course, with enough care all such situations can be dealt with, but that is beside the point. It is worth sacrificing some amount of the generality of algebraic effects and monads in exchange for predictable and safe usage of external computational effects, so long as the vast majority of common use cases are accommodated.

\subsubsection*{Contributions}

We address the described issues by showing how to design a programming
language based on \emph{runners of algebraic effects}. We
review runners in~\cref{sect:runnersbyexample} and recast them as
a programming construct in~\cref{sec:programming-with-runners}.
In \cref{sect:corecalculus}, we present $\lambdacoop$, a calculus that captures the core ideas of programming with
runners. 
We provide a coherent and sound denotational semantics for $\lambdacoop$ in \cref{sec:denotational-semantics}, where we also prove that well-typed code is properly finalised.
In~\cref{sect:examples}, we show examples of runners in action. The paper is accompanied by a
prototype language \pl{Coop} and a \pl{Haskell} library \pl{Haskell-Coop}, based on~$\lambdacoop$, see~\cref{sect:implementation}. The relationship between $\lambdacoop$ and existing work is addressed in~\cref{sect:relatedwork}, and future possibilities discussed in \cref{sect:conclusion}.

The paper is also accompanied by an online appendix (\url{https://arxiv.org/abs/1910.11629}) that provides the typing and equational rules we omit in \cref{sect:corecalculus}.

Runners are \emph{modular} in that they
can be used not only to model the top-level interaction
with the external environment, but programmers can also use them to
define and nest their own intermediate ``virtual machines''.
Our runners are \emph{effectful}: they may handle operations by calling further outer operations, and raise exceptions and send signals, through which  exceptional conditions and runtime errors are communicated back to user programs in a safe
fashion that preserves linear usage of external resources and ensures their proper finalisation.

We achieve \emph{suitable generality} for handling of external resources by
showing how runners provide implementations of algebraic operations together with
a natural notion of finalisation, and a strong guarantee that in the absence of
external kill signals the finalisation code is executed exactly once (\cref{thm:finalisation}). We argue
that for most purposes such discipline is well worth having, and giving up the
arbitrariness of effect handlers is an acceptable price to pay.
In fact, as will be apparent in the denotational semantics, runners are simply
a restricted form of handlers, which apply the continuation at most once in a tail call position.

Runners guarantee \emph{linear usage of resources} not through a linear or uniqueness type system (such as in the \pl{Clean} programming language~\cite{Koopman:FPinClean}) or a syntactic discipline governing the application of continuations in handlers, but rather by a design based on the linear state-passing technique studied by Møgelberg and Staton~\cite{Mogelberg:LinearUsageOfState}.
In this approach, a computational resource may be implemented without restrictions, but is then guaranteed to be used linearly by user code.



\section{Algebraic effects, handlers, and runners}
\label{sect:runnersbyexample}

We begin with a short overview of the theory of algebraic effects and handlers, as well as
runners. To keep focus on how runners give rise to a programming
concept, we work naively in set theory. Nevertheless, we use
category-theoretic language as appropriate, to make it clear that there are no essential
obstacles to extending our work to other settings (we return to this point in \cref{sec:semantics-types}).

\subsection{Algebraic effects and handlers}
\label{sect:algebraiceffects}

There is by now no lack of material on the algebraic approach to structuring computational effects.
For an introductory treatment we refer to~\cite{Bauer:WhatIsAlgebraic}, while of course
also recommend the seminal papers by Plotkin and
Power~\cite{Plotkin:SemanticsForAlgOperations,Plotkin:NotionsOfComputation}. The brief
summary given here only recalls the essentials and introduces notation.

An \emph{(algebraic) signature} is given by a set $\sig$ of \emph{operation symbols},
and for each $\op \in \sig$ its \emph{operation signature}
$
  \op : A_\op \opto B_\op
$,
where  $A_\op$ and $B_\op$ are called the \emph{parameter} and \emph{arity} set.
A \emph{$\sig$-structure} $\M$ is given by a carrier set $\Mcarrier$, and for each
operation symbol $\op \in \sig$, a map $\op_{\M} : A_\op \times (B_\op \expto \Mcarrier) \to \Mcarrier$,
where~$\expto$ is set exponentiation. The \emph{free $\sig$-structure~$\Tree{\sig}{X}$}
over a set~$X$ is the set of well-founded trees generated inductively by
\begin{itemize}
\item $\retTree{x} \in \Tree{\sig}{X}$, for every $x \in X$, and
\item $\op(a, \kappa) \in \Tree{\sig}{X}$, for every $\op \in \sig$, $a \in A_\op$, and $\kappa : B_\op \to \Tree{\sig}{X}$.
\end{itemize}
We are abusing notation in a slight but standard way, by using $\op$ both as the
name of an operation and a tree-forming constructor.
The elements of $\Tree{\sig}{X}$ are called \emph{computation trees}: a leaf
$\retTree{x}$ represents a pure computation returning a value $x$, while
$\op(a, \kappa)$ represents an effectful computation that calls $\op$ with
parameter~$a$ and continuation~$\kappa$, which expects a result from~$B_\op$.

An \emph{algebraic theory $\Th = (\Thsig, \Theq)$} is given by a \emph{signature~$\Thsig$} and
a set of \emph{equations~$\Theq$}.
The equations $\Theq$ express computational behaviour via interactions between
operations, and
are written in a suitable formalism, e.g.,~\cite{Plotkin:HandlingEffects}. We
explain these by way of examples, as the precise details do not matter for our purposes.
Let $\Zero = \{\,\}$ be the empty set and $\One = \{\star\}$ the standard singleton.

\begin{example}
  \label{ex:state}
  Given a set $C$ of possible states, the theory of \emph{$C$-valued state} has two operations, whose
  somewhat unusual naming will become clear later on,
  \begin{equation*}
    \siggetenv : \One \opto C,
    \qquad\qquad
    \sigsetenv : C \opto \One
  \end{equation*}
  and the equations (where we elide appearances of $\star$):
  \begin{gather*}
    \siggetenv (\lam{c} \sigsetenv(c, \kappa)) = \kappa,
    \qquad
    \sigsetenv (c, \siggetenv\, \kappa) = \sigsetenv(c, \kappa\, c), \\
    \sigsetenv (c, \sigsetenv(c', \kappa)) = \sigsetenv(c', \kappa).
  \end{gather*}
  For example,
  the second equation states that reading state right after setting it to~$c$ gives precisely~$c$.
  The third equation states that $\sigsetenv$ overwrites the state.
\end{example}

\begin{example}
  \label{ex:exceptions}
  Given a set of exceptions $E$, the algebraic theory of \emph{$E$-many exceptions} is given by
  a single operation $\sigraise : E \opto \Zero$, and no equations.
\end{example}

A \emph{$\Th$-model}, also called a \emph{$\Th$-algebra}, is a $\Thsig$-structure which
satisfies the equations in $\Theq$. The \emph{free $\Th$-model} over a set~$X$ is constructed
as the quotient
\begin{equation*}
  \FreeAlg{\Th}{X} = \Tree{\Thsig}{X}/{\sim}
\end{equation*}
by the $\Thsig$-congruence $\sim$ generated by $\Theq$. Each $\op \in \Thsig$
is interpreted in the free model as the map
$(a, \kappa) \mapsto [\op(a, \kappa)]$, where $[{-}]$ is the $\sim$-equivalence class.

$\FreeAlg{\Th}{-}$ is the functor part of a \emph{monad} on sets, whose \emph{unit} at a
set~$X$ is
\begin{equation*}
  \xymatrix@C=3em@R=2em@M=0.5em{
    {X} \ar[r]^(0.35){\retTree{}}
    &
    {\Tree{\Thsig}{X}} \ar@{->>}[r]^{[{-}]}
    &
    {\FreeAlg{\Th}{X}.}
  }
\end{equation*}
The \emph{Kleisli extension} for this monad is then the operation which lifts any map \linebreak
$f : X \to \Tree{\Thsig}{Y}$ to the map $\lift{f} : \FreeAlg{\Thsig}{X} \to \FreeAlg{\Thsig}{Y}$,
given by
\begin{equation*}
  \lift{f}\,[\retTree{x}] \defeq f \, x,
  \qquad\qquad
  \lift{f}\,[\op(a, \kappa)] \defeq [\op(a, \lift{f} \circ \kappa)].
\end{equation*}
That is, $\lift{f}$ traverses a computation tree and replaces each leaf $\retTree{x}$
with $f\,x$.

The preceding construction of free models and the monad may be retro-fitted to an
algebraic signature~$\sig$, if we construe~$\sig$ as an algebraic theory with no
equations. In this case~$\sim$ is just equality, and so we may omit the quotient and the
pesky equivalence classes. Thus the carrier of the free $\sig$-model is the set of
well-founded trees $\Tree{\sig}{X}$, with the evident monad structure.

A fundamental insight of Plotkin and
Power~\cite{Plotkin:SemanticsForAlgOperations,Plotkin:NotionsOfComputation} was that many
computational effects may be adequately described by algebraic theories, with
the elements of free models corresponding to effectful computations. For example, the monads
induced by the theories from \cref{ex:state,ex:exceptions} are respectively isomorphic to the usual \emph{state monad}
$\St{C}\,X \defeq (C \Rightarrow X \times C)$ and the \emph{exceptions monad}
$\Exc{E}\,X \defeq X + E$.

Plotkin and Pretnar~\cite{Plotkin:HandlingEffects} further
observed that the universal property of free models may be used to model a
programming concept known as \emph{handlers}. Given a $\Th$-model $\M$ and a map
$f : X \to \Mcarrier$, the universal property of the free $\Th$-model gives us a
unique $\Th$-homomorphism $\freelift{f} : \FreeAlg{\Th}{X} \to \Mcarrier$ satisfying
\begin{equation*}
  \freelift{f} \, [\retTree{x}] = f\,x,
  \qquad\qquad
  \freelift{f} \, [\op(a, \kappa)] = \op_\M(a, \freelift{f} \circ \kappa).
\end{equation*}

A handler for a theory $\Th$ in a language such as \pl{Eff} amounts to a model 
$\M$ whose carrier $\Mcarrier$
is the carrier $\FreeAlg{\Th'}{Y}$ of the free model
for some other theory~$\Th'$, while the associated handling
construct is the induced $\Th$-homomorphism $\FreeAlg{\Th}{X} \to \FreeAlg{\Th'}{Y}$.
Thus handling transforms computations with effects~$\Th$ to computations with
effects~$\Th'$. There is however no restriction on how a handler implements an
operation, in particular, it may use its continuation in an arbitrary
fashion.
We shall put the universal property of free models to good use as well, while
making sure that the continuations are always used affinely.

\subsection{Runners}
\label{sect:purerunners}

Much like monads, handlers are useful for simulating computational effects, because they
allow us to transform $\Th$-computations to $\Th'$-computations. However, eventually there
has to be a ``top level'' where such transformations cease and actual computational
effects happen. For these we need another concept, known as
\emph{runners}~\cite{Uustalu:Runners}.
Runners are equivalent to the concept of
\emph{comodels}~\cite{Plotkin:TensorsOfModels,Power:Comodels}, which are ``just
models in the opposite category'', although one has to apply the motto
correctly by using powers and co-powers where seemingly exponentials and products would
do. 
Without getting into the intricacies, let us spell out the definition.

\begin{definition}
  A \emph{runner} $\R$ for a signature $\sig$ is given by a carrier set~$\Rcarrier$ together with, for
  each $\op \in \sig$, a
  \emph{co-operation~$\coop_{\R} : A_\op \to (\Rcarrier \expto B_\op \times \Rcarrier)$.}
\end{definition}

Runners are usually defined to have co-operations in the equivalent uncurried form
$\coop_\R : A_\op \times \Rcarrier \to B_\op \times \Rcarrier$, but that is less convenient for our purposes.

Runners may be defined more generally for theories $\Th$, rather than just signatures,
by requiring that the co-operations satisfy $\Theq$. We shall have no use for these,
although we expect no obstacles in incorporating them into our work.

A runner tells us what to do when an effectful computation reaches the top-level runtime
environment. Think of~$\Rcarrier$ as the set of configurations of the runtime environment. Given
the current configuration $c \in \Rcarrier$, the operation $\op(a, \kappa)$ is executed as the
corresponding co-operation $\coop_\R\,a\,c$ whose result $(b, c') \in B_\op \times \Rcarrier$ gives
the result of the operation $b$ and 
the next runtime configuration $c'$. The continuation $\kappa\,b$
then proceeds in runtime configuration~$c'$.

It is not too difficult to turn this idea into a mathematical model. For any
$X$, the co-operations induce a $\sig$-structure $\M$ with
$\Mcarrier \defeq \St{\Rcarrier} X = (\Rcarrier \expto X \times \Rcarrier)$ 
and operations $\op_\M : A_\op \times (B_\op \expto \St{\Rcarrier} X) \to \St{\Rcarrier} X$
given by
\begin{equation*}
  \op_\M (a, \kappa) \defeq \lam{c} \kappa\, (\pi_1 (\coop_\R\,a\,c))\, (\pi_2 (\coop_\R\,a\,c)).
\end{equation*}
We may then use the universal property of the free $\sig$-model to obtain a $\sig$-homomorphism
$\runh_X : \Tree{\sig}{X} \to \St{\Rcarrier} X$ satisfying the equations
\[
  \runh_X(\retTree{x}) = \lam{c} (x, c),
  \qquad\qquad
  \runh_X(\op(a, \kappa)) = \op_\M(a, \runh_X \circ \kappa).
\]
The map $\runh_X$ precisely captures the idea that a runner
\emph{runs computations} by transforming (static) computation trees into
state-passing maps. Note how in the above definition of $\op_\M$, the
continuation~$\kappa$ is used in a controlled way, as it appears precisely once
as the head of the outermost application. In terms of programming, this
corresponds to linear use in a tail-call position.

Runners are less ad-hoc than they may seem. First, notice that $\op_\M$ is just the
composition of the co-operation $\coop_\R$ with
the state monad's Kleisli extension of the continuation $\kappa$, and so is
the standard way of turning \emph{generic effects} into $\sig$-structures~\cite{Plotkin:AlgOperations}.
Second, the map $\runh_X$ is the component at $X$ of a monad morphism
$\runh : \Tree{\sig}{{-}} \to \St{\Rcarrier}$. Møgelberg \& Staton~\cite{Mogelberg:LinearUsageOfState}, as
well as Uustalu~\cite{Uustalu:Runners}, showed that the passage from a runner~$\R$ to the
corresponding monad morphism~$\runh$ forms a one-to-one correspondence between the former and the
latter.

As defined, runners are too restrictive a model of top-level computation, because the only
effect available to co-operations is state, but in practice the runtime
environment may also signal errors and perform other effects, by calling its own runtime
environment. We are led to the following generalisation.

\begin{definition}
  For a signature $\sig$ and monad $\T$, a \emph{$\T$-runner $\R$} for~$\sig$, 
  or just an \emph{effectful runner}, 
  is given by, for each $\op \in \sig$, a \emph{co-operation}
  $\coop_\R : A_\op \to \T B_\op$.
\end{definition}

The correspondence between runners and monad morphisms still holds.

\begin{proposition}
  \label{prop:monadmorphism}
  For a signature $\sig$ and a monad $\T$, the monad morphisms $\Tree{\sig}{{-}} \to \T$
  are in one-to-one correspondence with $\T$-runners for~$\sig$.
\end{proposition}

\begin{proof}
  This is an easy generalisation of the correspondence for 
  ordinary runners. Let us fix a signature $\sig$, and a monad $\T$ 
  with unit $\eta$ and Kleisli extension $\lift{{-}}$.

  Let $\R$ be a $\T$-runner for $\sig$. For any set $X$, $\R$ induces a $\sig$-structure
  $\M$ with $\Mcarrier \defeq \T X$ and 
  $\op_\M : A_\op \times (B_\op \expto \T X) \to \T X$ defined as
  $
    \op_\M (a, \kappa) \defeq \lift{\kappa} (\coop_R\,a)
  $.
  As before, the universal property of the free model $\Tree{\sig}{X}$ provides a unique
  $\sig$-homomorphism $\runh_X : \Tree{\sig}{X} \to \T X$, satisfying the equations
  \begin{equation*}
    \runh_X (\retTree{x}) = \eta_X(x),
    \qquad\qquad
    \runh_X (\op(a, \kappa)) = \op_\M (a, \runh_X \circ \kappa).
  \end{equation*}
  The maps $\runh_X$ collectively give us the desired monad morphism $\runh$ induced by $\R$.
  
  Conversely, given a monad morphism $\theta : \Tree{\sig}{{-}} \to \T$, we may recover a
  $\T$-runner~$\R$ for $\sig$ by defining the co-operations as
  $
    \coop_\R \, a \defeq \theta_{B_\op} (\op (a, \lam{b} \retTree{b}))
  $.
  It is not hard to check that we have described a one-to-one correspondence.
  \qed
\end{proof}



\section{Programming with runners}
\label{sec:programming-with-runners}

If ordinary runners are not general enough, the effectful ones are too general: 
parameterised by arbitrary monads $\T$, they do
not combine easily and they lack a clear notion of resource management. Thus,
we now engineer more specific monads whose associated runners can be turned into a
programming concept.
While we give up complete generality, the monads presented below are still quite
versatile, as they are parameterised by arbitrary algebraic signatures $\Sigma$,
and so are extensible and support various combinations of effects.

\subsection{The user and kernel monads}
\label{sec:user-kernel-monads}

Effectful source code running inside a runtime environment is just one example of a more
general phenomenon in which effectful computations are enveloped by a layer that provides
a supervised access to external resources: a user process is controlled by a kernel, a web
page by a browser, an operating system by hardware, or a virtual machine, etc. We shall
adopt the parlance of software systems, and refer to the two layers generically as the
\emph{user} and \emph{kernel} code.
Since the two kinds of code need not, and will not, use the same effects, each
will be described by its own algebraic theory and compute in its own monad.

We first address the kernel theory. 
Specifically, we look for an algebraic theory such that effectful runners for the induced monad
satisfy the following desiderata:
\begin{enumerate}
\item Runners support management and controlled finalisation of resources.
\item Runners may use further external resources.
\item Runners may signal failure caused by unavoidable circumstances.
\end{enumerate}

The totality of external resources
available to user code appears as a stateful external environment, even though it
has no direct access to it. Thus, kernel computations should carry state. We
achieve this by incorporating into the kernel theory the operations $\siggetenv$
and $\sigsetenv$, and equations for state from \cref{ex:state}.

Apart from managing state, kernel code should have access to further
effects, which may be true external effects, or some outer
layer of runners. In either case, we should allow the kernel code to call
operations from a given signature~$\sig$.

Because kernel computations ought to be able to signal failure, we should
include an exception mechanism. In practice, many programming languages and
systems have two flavours of exceptions, variously called recoverable and fatal,
checked and unchecked, exceptions and errors, etc. One kind, which we call just
\emph{exceptions}, is raised by kernel code when a situation requires special
attention by user code. The other kind, which we call \emph{signals},
indicates an unrecoverable condition that prevents normal execution of user
code. These correspond precisely to the two standard ways of combining
exceptions with state, namely the coproduct and the tensor of algebraic
theories~\cite{Hyland:CombiningEffects}. The coproduct simply adjoins exceptions
$\sigraise : E \leadsto \Zero$ from \cref{ex:exceptions} to the theory of
state, while the tensor extends the theory of state with signals
$\sigkill : S \leadsto \Zero$, together with equations
\begin{equation}
  \label{eq:kill-state}%
  \siggetenv(\lam{c} \sigkill\,s) = \sigkill\,s,
  \qquad\qquad
  \sigsetenv(c, \sigkill\,s) = \sigkill\,s.
\end{equation}
These equations say that a signal discards state, which makes it unrecoverable.

To summarise, the \emph{kernel theory} $\ThKK{C}{\sig}{E}{S}$ contains 
operations from a signature $\sig$, as well as state operations
$\siggetenv : \One \opto C$, $\sigsetenv : C \opto \One$, exceptions
$\sigraise : E \opto \Zero$, and signals $\sigkill : S \opto \Zero$, with equations for state
from \cref{ex:state}, equations~\eqref{eq:kill-state} relating state and
signals, and for each operation $\op \in \sig$, equations
\begin{align*}
  \siggetenv(\lam{c} \op(a, \kappa\,c)) &= \op(a, \lam{b} \siggetenv (\lam{c} \kappa\,c\,b)),\\
  \sigsetenv(c, \op(a, \kappa)) &= \op(a, \lam{b} \sigsetenv(c, \kappa\,b)),
\end{align*}
expressing that external operations do not interact with kernel state. 
It is not difficult to see that $\ThKK{C}{\sig}{E}{S}$ induces, up to
isomorphism, the \emph{kernel monad}
\begin{equation*}
  \KK{C}{\sig}{E}{S} X \quad\defeq\quad C \expto \Tree{\sig}{((X + E) \times C) + S}.
\end{equation*}

How about user code? It can of course call operations from a 
signature~$\sig$ (not necessarily the same as the kernel code), and because we
intend it to handle exceptions, it might as well have the ability to raise them.
However, user code knows nothing about signals and kernel state. Thus, we choose the \emph{user theory
  $\ThUU{\sig}{E}$} to be the algebraic theory with operations $\sig$, exceptions
$\sigraise : E \opto \Zero$, and no equations. This theory induces the \emph{user
  monad} $\UU{\sig}{E} X \defeq \Tree{\sig}{X + E}$.

\subsection{Runners as a programming construct}
\label{sec:runn-as-progr}

In this section, we turn the ideas presented so far into programming constructs.
We strive for a realistic result,
but when faced with several design options, we prefer simplicity and semantic
clarity. We focus here on translating the central concepts, and postpone
various details to \cref{sect:corecalculus}, where we present a full calculus.

We codify the idea of user and kernel computations by having syntactic
categories for each of them, as well as one for values. We use letters $M$, 
$N$ to indicate user computations, $K$, $L$ for kernel computations, 
and $V$, $W$ for values.

User and kernel code raise exceptions with operation $\tmkw{raise}$, and catch
them with exception handlers based on Benton and Kennedy's \emph{exceptional
  syntax}~\cite{Benton:ExceptionalSyntax},
\begin{equation*}
  \tmtry{M}{\{
    \tmreturn{x} \mapsto N,
    \ldots, \tmraise{e} \mapsto N_e, \ldots
  \}}, 
\end{equation*}
and analogously for kernel code. The familiar binding construct
$\tmlet{x}{M}{N}$
is simply shorthand for
$\tmtry{M}{\{\tmreturn{x} \mapsto N, \ldots, \tmraise{e} \mapsto \tmraise{e}, \ldots\}}$.

As a programming concept, a runner $R$ takes the form
\begin{equation*}
  \tmrunner{(\tm{op}\,x \mapsto K_{\tm{op}})_{\tm{op} \in \sig}}{C}, 
\end{equation*}
where each $K_\op$ is a kernel computation, with the variable $x$ bound in $K_{\tm{op}}$, so that
each clause $\tm{op} \, x \mapsto K_{\tm{op}}$ determines a co-operation for the
kernel monad. The subscript $C$ indicates the type of the state used by  
the kernel code $K_\op$.

The corresponding elimination form is a handling-like construct
\begin{equation}
  \label{eq:using}
  \tmrun{R}{V}{M}{F}, 
\end{equation}
which uses the co-operations of runner $R$ ``at'' initial kernel state~$V$ to
run user code~$M$, and finalises its return value, exceptions, and signals
with~$F$, see~\eqref{eq:finally-clause} below.
When user code $M$ calls an operation $\op$, the enveloping $\tmkw{run}$ construct runs the
corresponding co-operation $K_\op$ of $R$. While doing so, $K_\op$ might raise 
exceptions. But not every exception makes sense for every operation, and so
we assign to each operation $\op$ a set of exceptions $E_\op$ which the
co-operations implementing it may raise, by augmenting its operation signature with $E_\op$, as 
\begin{equation*}
  \op : \tysigop{A_\op}{B_\op}{E_\op}.
\end{equation*}
An exception raised by the co-operation $K_\op$ propagates back to the operation call in
the user code. Therefore, an operation call should have not only a continuation $x\,.\,M$
receiving a result, but also continuations $N_e$, one for each $e \in E_\op$,
\begin{equation*}
  \tmop{op}{}{V}{\tmcont x M}{\tmexccont N e {E_\op}}.
\end{equation*}
If $K_\op$ returns a value $b \in B_\op$, the execution proceeds 
as $M[b/x]$, and as $N_e$ if $K_\op$ raises an exception 
$e \in E_\op$. In examples, we use the generic versions of 
operations~\cite{Plotkin:AlgOperations}, written $\tmgeneff{op}{V}$,
which pass on return values and re-raise exceptions.

One can pass exceptions back to operation calls also in 
a language with handlers, such as \pl{Eff}, by 
changing the signatures of operations to
$A_\op \opto B_\op + E_\op$, and implementing 
the exception mechanism by hand, so that every
operation call is followed by a case distinction on $B_\op + E_\op$.
One is reminded of how operating system calls communicate 
errors back to user code as exceptional values.

A co-operation $K_\op$ may also send a signal, in which case the rest of the user code $M$
is skipped and the control proceeds directly to the corresponding case of the finalisation part~$F$ of the
$\tmkw{run}$ construct~\eqref{eq:using}, whose syntactic form is
\begin{equation}
  \label{eq:finally-clause}%
  \{ \tmreturn{x} \at c \mapsto N,
     \ldots, \tmraise{e} \at c \mapsto N_e,
     \ldots, \tmkill{s} \mapsto N_s,
     \ldots
  \}.
\end{equation}
Specifically, if $M$ returns a value $v$, then $N$ is evaluated with $x$ bound
to $v$ and $c$ to the final kernel state; if~$M$ raises an
exception~$e$ (either directly or indirectly via a co-operation of $R$), 
then $N_e$ is executed, again with $c$ bound to the final kernel state; and 
if a co-operation of $R$ sends a signal $s$, then $N_s$ is executed.

\begin{example}
  \label{ex:file-IO}%
  In anticipation of setting up the complete calculus we show how one can work with files.
  The language implementors can provide an operation $\tm{open}$ which opens a
  file for writing and returns its file handle, an operation $\tm{close}$ which closes a
  file handle, and a runner $\mathsf{fileIO}$ that implements
  writing.
  Let us further suppose that $\mathsf{fileIO}$ may raise
  an exception $\mathsf{QuotaExceeded}$ if a write exceeds the user disk quota,
  and send a signal $\mathsf{IOError}$ 
  if an unrecoverable external error occurs.
  The following code illustrates how to guarantee proper closing of the file:
\begin{lstlisting}
using fileIO @ (open "hello.txt") run
  write "Hello, world."
finally {
  return x @ fh -> close fh,
  raise QuotaExceeded @ fh -> close fh,
  kill IOError -> return () }
\end{lstlisting}
  Notice that the user code does not have direct access to the file handle.
  Instead, the runner holds it in its state, where it is available to the
  co-operation that implements $\tm{write}$. The finalisation block gets access to
  the file handle upon successful completion and raised exception, so it can close
  the file, but when a signal happens the finalisation cannot close the file,
  nor should it attempt to do so.

  We also mention that the code ``cheats'' by placing the call to $\tm{open}$ in a
  position where a value is expected. We should have $\tmkw{let}$-bound the file handle
  returned by $\tm{open}$ outside the $\tmkw{run}$ construct, which would make it clear that
  opening the file happens \emph{before} this construct (and that $\tm{open}$ is
  \emph{not} handled by the finalisation), but would also expose the file handle. Since
  there are clear advantages to keeping the file handle inaccessible, a realistic
  language should accept the above code and hoist computations from value positions
  automatically.
\end{example}



\newcommand{\tmcoop}[3]{\mathtt{\overline{#1}}~#2 \mapsto #3}

\section{A calculus for programming with runners}
\label{sect:corecalculus}

Inspired by the semantic notion of runners and the ideas of the previous
section, we now present a calculus for programming with co-operations and
runners, called $\lambdacoop$. It is a low-level fine-grain call-by-value
calculus~\cite{Levy:CBPV}, and as such could inspire an intermediate language
that a high-level language is compiled to.

\subsection{Types}
\label{sect:types}

\begin{figure}[tb]
  \parbox{\textwidth}{
  \centering
  \small
  \begin{align*}
  \text{Ground type $A$, $B$, $C$}
  \bnfis& \tybase      & &\text{base type} \\
  \bnfor& \tyunit       & &\text{unit type} \\
  \bnfor& \tyempty      & &\text{empty type} \\
  \bnfor& \typrod{A}{B} & &\text{product type} \\
  \bnfor& \tysum{A}{B}  & &\text{sum type}
  \\[1ex]
  \text{Constant signature:}
  \phantom{\bnfis}& \tmconst{f} : (A_1,\ldots,A_n) \to B
  \\[1ex]
  \text{Signature $\sig$}
  \bnfis& \{\op_1, \op_2, \ldots, \op_n\} \subset \Ops
  \\
  \text{Exception set $E$}
  \bnfis& \{e_1, e_2, \ldots, e_n\} \subset \Excs
  \\
  \text{Signal set $S$}
  \bnfis& \{s_1, s_2, \ldots, s_n\} \subset \Sigs
  \\[1ex]
  \text{Operation signature:}
  \phantom{\bnfis}& \op : \tysigop{A_\op}{B_\op}{E_\op}
  \\[1ex]
  \text{Value type $X$, $Y$, $Z$}
  \bnfis& A                                       & &\text{ground type} \\
  \bnfor& \typrod{X}{Y}                           & &\text{product type} \\
  \bnfor& \tysum{X}{Y}                            & &\text{sum type} \\
  \bnfor& \tyfun{X}{\tyuser{Y}{\Ueff}}           & &\text{user function type} \\
  \bnfor& \tyfunK{X}{\tykernel{Y}{\Keff}}         & &\text{kernel function type} \\
  \bnfor& \tyrunner{\sig}{\sig'}{S}{C}      & &\text{runner type}
  \\[1ex]
  \text{User (computation) type:}
  \phantom{\bnfor} &\tyuser{X}{\Ueff} \quad \text{where $\Ueff = (\sig, E)$}
  \\
  \text{Kernel (computation) type:}
  \phantom{\bnfor}& \tykernel{X}{\Keff} \quad \text{where $\Keff = (\sig, E, S, C)$}
  \end{align*}
  } 
  \caption{The types of $\lambdacoop$.}
  \label{fig:lambdacoop-types}
\end{figure}

The types of $\lambdacoop$ are shown in \cref{fig:lambdacoop-types}.
The \emph{ground types} contain \emph{base types}, and are closed under finite sums and
products. These are used in operation signatures and as types of kernel state. (Allowing
arbitrary types in either of these entails substantial complications that can be dealt
with but are tangential to our goals.) Ground types can also come with corresponding 
constant symbols~$\tmconst{f}$, each associated with a fixed \emph{constant signature}
$\tmconst{f} : (A_1,\ldots,A_n) \to B$.

We assume a supply of operation symbols $\Ops$, exception names $\Excs$, and
signal names $\Sigs$. Each operation symbol~$\op \in \Ops$ is equipped with an
\emph{operation signature} $\tysigop{A_\op}{B_\op}{E_\op}$, which specifies its
parameter type~$A_\op$ and arity type~$B_\op$, and the exceptions~$E_\op$ 
that the corresponding co-operations may raise in runners.

The \emph{value types} extend ground types with two 
function types, and a type of runners.
The \emph{user function type $\tyfun{X}{\tyuser{Y}{(\sig, E)}}$} classifies
functions taking arguments of type~$X$ to computations classified by the \emph{user
  (computation) type}~${\tyuser{Y}{(\sig, E)}}$, i.e., those that return values of
type~$Y$, and may call operations~$\sig$ and raise exceptions~$E$.
Similarly, the \emph{kernel function type
  $\tyfunK{X}{\tykernel{Y}{(\sig, E, S, C)}}$} classifies functions taking
arguments of type~$X$ to computations classified by the \emph{kernel
  (computation) type~$\tykernel{Y}{(\sig, E, S, C)}$}, i.e., those that return
values of type~$Y$, and may call operations~$\sig$, raise exceptions~$E$, send
signals~$S$, and use state of type~$C$. We note that the ingredients for user and kernel types
correspond precisely to the parameters of the user monad $\UU{\sig}{E}$ and the
kernel monad $\KK{C}{\sig}{E}{S}$ from \cref{sec:user-kernel-monads}.
Finally, the \emph{runner type $\tyrunner{\sig}{\sig'}{S}{C}$} classifies runners that
implement co-operations for the operations~$\sig$ as kernel computations which use 
operations~$\sig'$, send signals~$S$, and use state of type~$C$.

\subsection{Values and computations}
\label{sec:values-computations}

\begin{figure}[tp]
  \parbox{\textwidth}{
  \centering
  \small
  \abovedisplayskip=0pt
  \begin{align*}
  \intertext{\textbf{Values}}
  V, W
  \bnfis& x                                       & &\text{variable} \\
  \bnfor& \tmconst{f}(V_1,\ldots,V_n)                                       & &\text{ground constant} \\
  \bnfor& \tmunit                                 & &\text{unit} \\
  \bnfor& \tmpair{V}{W}                           & &\text{pair} \\
  \bnfor& \tminl[X,Y]{V} \bnfor \tminr[X,Y]{V}    & &\text{injection} \\
  \bnfor& \tmfun{x : X}{M}                        & &\text{user function} \\
  \bnfor& \tmfunK{x : X}{K}                       & &\text{kernel function} \\
  \bnfor& \tmrunner{(\tm{op}\,x \mapsto K_{\tm{op}})_{\tm{op} \in \sig}}{C}
                                                  & &\text{runner}
  \\[1ex]
  \intertext{\textbf{User computations}}
  M, N
  \bnfis& \tmreturn{V}                            & &\text{value} \\
  \bnfor& V\,W                                    & &\text{application} \\
  \bnfor& \tmtry{M}{
          \{ \tmreturn{x} \mapsto N,
             (\tmraise{e} \mapsto N_e)_{e \in E} \}
          }
                                                  & &\text{exception handler} \\
  \bnfor& \tmmatch{V}{\tmpair{x}{y} \mapsto M}    & &\text{product elimination} \\
  \bnfor& \tmmatch[X]{V}{}                        & &\text{empty elimination} \\
  \bnfor& \tmmatch{V}{\tminl{x} \mapsto M, \tminr{y} \mapsto N}
                                                  & &\text{sum elimination} \\
  \bnfor& \tmop{op}{X}{V}{\tmcont x M}{\tmexccont N e {E_\op}}
                                                  & &\text{operation call} \\
  \bnfor& \tmraise[X]{e}                          & &\text{raise exception} \\
  \bnfor& \tmrun{V}{W}{M}{F} 
                                                  & &\text{running user code} \\
  \bnfor& 
            \tmkernel{K}{W}{F}
                                                  & &\text{switch to kernel mode}
  \\[2ex]
  F \bnfis & \omit \rlap{$\{ \tmreturn{x} \at c \mapsto N, 
                    (\tmraise{e} \at c \mapsto N_e)_{e \in E},
                    (\tmkill{s} \mapsto N_s)_{s \in S} \}$}
  \\[1ex]
  \intertext{\textbf{Kernel computations}}
  K, L
  \bnfis& \tmreturn[C]{V}                         & &\text{value} \\
  \bnfor& V\,W                                    & &\text{application} \\
  \bnfor& \tmtry{K}{
          \{ \tmreturn{x} \mapsto L,
             (\tmraise{e} \mapsto L_e)_{e \in E} \}
          }
                                                  & &\text{exception handler} \\
  \bnfor& \tmmatch{V}{\tmpair{x}{y} \mapsto K}    & &\text{product elimination} \\
  \bnfor& \tmmatch[X \at C]{V}{}                  & &\text{empty elimination} \\
  \bnfor& \tmmatch{V}{\tminl{x} \mapsto K, \tminr{y} \mapsto L}
                                                  & &\text{sum elimination} \\
  \bnfor& \tmop{op}{X}{V}{\tmcont x K}{\tmexccont L e {E_\op}}
                                                  & &\text{operation call} \\
  \bnfor& \tmraise[X \at C]{e}                    & &\text{raise exception} \\
  \bnfor& \tmkill[X \at C]{s}                     & &\text{send signal} \\
  \bnfor& \tmgetenv[C]{\tmcont c K}               & &\text{get kernel state} \\
  \bnfor& \tmsetenv{V}{K}                         & &\text{set kernel state} \\
  \bnfor& \tmuser{M}{
          \begin{aligned}[t]
          \{ &\tmreturn{x} \mapsto K,
             (\tmraise{e} \mapsto L_e)_{e \in E} \}
          \end{aligned}
          }
                                                  & &\text{switch to user mode}
  \end{align*}
  } 
  \caption{Values, user computations, and kernel computations of $\lambdacoop$.}
  \label{fig:lambdacoop-terms}
\end{figure}

The syntax of terms is shown in \cref{fig:lambdacoop-terms}. The
usual fine-grain call-by-value stratification of terms into pure values and effectful
computations is present, except that we further distinguish between \emph{user} and
\emph{kernel} computations.

\subsubsection{Values}
\label{sec:values}

Among the values are variables, constants for ground types, and constructors 
for sums and products. There are two kinds of functions, for abstracting over user and
kernel computations. A \emph{runner} is a value of the form
\begin{equation*}
  \tmrunner{(\tm{op}\,x \mapsto K_{\tm{op}})_{\tm{op} \in \sig}}{C}.
\end{equation*}
It implements co-operations for operations $\tm{\op}$ as kernel
computations~$K_\op$, with $x$ bound in~$K_\op$. The type annotation~$C$
specifies the type of the state that~$K_\op$ uses.
Note that $C$ ranges over ground types, a restriction that allows us to define a naive
set-theoretic semantics.
We sometimes omit these type annotations.

\subsubsection{User and kernel computations}

The user and kernel computations both have pure computations, function application,
exception raising and handling, standard elimination forms, and operation calls.
Note that the typing annotations on some of these differ according to their mode.
For instance, a user operation call is annotated
with the result type~$X$, whereas the annotation $X \at C$ on a kernel operation call
also specifies the kernel state type~$C$.

The binding construct $\tmlet[X ! E]{x}{M}{N}$ is not part of the syntax,
but is an abbreviation for
$
  \tmtry{M}{
  \{ \tmreturn{x} \mapsto N,
     (\tmraise{e} \mapsto \tmraise[X]{e})_{e \in E}
   \}}
$,
and there is an analogous one for kernel computations. We often drop the 
annotation $X ! E$.

Some computations are specific to one or the other mode. Only the kernel mode
may send a signal with $\tmkill{\!}$, and manipulate state with
$\tmkw{getenv}$ and $\tmkw{setenv}$, but only the user mode has the 
$\tmkw{run}$ construct from \cref{sec:runn-as-progr}.
Finally, each mode has the ability to ``context switch'' to the other one.
The kernel computation
\begin{equation*}
\tmuser{M}{
   \{
     \tmreturn{x} \mapsto K,
     (\tmraise{e} \mapsto L_e)_{e \in E}
   \}
}
\end{equation*}
runs a user computation $M$ and handles the returned value and leftover
exceptions with kernel computations $K$ and $L_e$.
Conversely, the user computation
\begin{equation*}
\tmkernel{K}{W}{
  \{x \at c \mapsto M,
    (\tmraise{e} \at c \mapsto N_e)_{e \in E},
    (\tmkill{s} \mapsto N_s)_{s \in S}
  \}
}
\end{equation*}
runs kernel computation $K$ with initial state $W$, and handles the returned value,
and leftover exceptions and signals with user computations $M$, $N_e$, and $N_s$.

\subsection{Type system}
\label{sec:typesystem}

We equip $\lambdacoop$ with a type system akin to type and effect systems for
algebraic effects and handlers~\cite{Bauer:EffectSystem,Benton:ExceptionalSyntax,Kammar:Handlers}.
We are experimenting with resource control, so it makes sense for the type system
to tightly control resources. Consequently, our effect system
does not allow effects to be implicitly propagated outwards.

In \cref{sect:types}, we assumed that each operation $\op \in \Ops$ 
is equipped with some fixed operation signature
$
  \op : \tysigop{A_\op}{B_\op}{E_\op}
$.
We also assumed a fixed constant signature $\tmconst{f} : (A_1, \ldots, A_n) \to B$
for each ground constant $\tmconst{f}$.
We consider this information to be part of the type system and say no more about it.

Values, user computations, and kernel computations each have a corresponding
\emph{typing judgement} form and a \emph{subtyping relation}, given by 
\begin{align*}
  &\Gamma \types V : X,
& &\Gamma \types M : \tyuser{X}{\Ueff},
& &\Gamma \types K : \tykernel{X}{\Keff},\\
  &X \sub Y,
& &\tyuser{X}{\Ueff} \sub \tyuser{Y}{\Veff},
& &\tykernel{X}{\Keff} \sub \tykernel{Y}{\Leff}, 
\end{align*}
where $\Gamma$ is a \emph{typing context} $x_1 : X_1, \ldots, x_n : X_n$.
The effect information is an over-approximation, i.e., $M$ and $K$ employ \emph{at
  most} the effects described by $\Ueff$ and $\Keff$.
The complete rules for these judgements are given in the online appendix. 
We comment here
only on the rules that are peculiar to~$\lambdacoop$, see \cref{fig:typing-selected}.

\begin{figure}[tp]
  \centering
  \small
  \begin{mathpar}
    \coopinfer{Sub-Ground}{ }{A \sub A}
    
    \coopinfer{Sub-Runner}{
      \sig_1' \subseteq \sig_1 \\
      \sig_2 \subseteq \sig_2' \\
      S \subseteq S' \\
      C \equiv C'
    }{
      \tyrunner{\sig_1}{\sig_2}{S}{C} \sub \tyrunner{\sig_1'}{\sig_2'}{S'}{C'}
    }

    \coopinfer{Sub-Kernel}{
      X \sub X' \\
      \sig \subseteq \sig' \\
      E \subseteq E' \\
      S \subseteq S' \\
      C \equiv C'
    }{
      \tykernel{X}{(\sig, E, S, C)} \sub \tykernel{X'}{(\sig', E', S', C')}
    }

  \coopinfer{TyUser-Try}{
    \Gamma \types M : \tyuser{X}{(\sig,E)}
    \\
    \Gamma, x \of X \types N : \tyuser{Y}{(\sig,E')}
    \\
    \big(
      \Gamma \types N_e : \tyuser{Y}{(\sig,E')}
    \big)_{e \in E}
  }{
    \Gamma \types
    \tmtry{M}{
        \{ \tmreturn{x} \mapsto N,
           (\tmraise{e} \mapsto N_e)_{e \in E} \}
        }
    : \tyuser{Y}{(\sig,E')}
  }

  \coopinfer{TyUser-Run}{
    F \equiv
    \{ \tmreturn{x} \at c \mapsto N,
       (\tmraise{e} \at c \mapsto N_e)_{e \in E},
       (\tmkill{s} \mapsto N_s)_{s \in S}
    \}
    \\\\
    \Gamma \types V : \tyrunner{\sig}{\sig'}{S}{C} \\
    \Gamma \types W : C \\\\
    \Gamma \types M : \tyuser{X}{(\sig, E)} \\
    \Gamma, x \of X, c \of C \types N : \tyuser{Y}{(\sig', E')} \\
    \big(
       \Gamma, c \of C \types N_e : \tyuser{Y}{(\sig', E')}
    \big)_{e \in E} \\
    \big(
       \Gamma \types N_s : \tyuser{Y}{(\sig', E')}
    \big)_{s \in S} \\
  }{
    \Gamma \types \tmrun{V}{W}{M}{F} : \tyuser{Y}{(\sig', E')}
  }

  \coopinfer{TyUser-Op}{
    \Ueff \equiv (\sig,E) \\
    \op \in \sig \\
    \Gamma \types V : A_\op \\\\
    \Gamma, x \of B_\op \types M : \tyuser{X}{\Ueff} \\
    \big(
      \Gamma \vdash N_e : \tyuser{X}{\Ueff}
    \big)_{e \in E_\op}
  }{
    \Gamma \types \tmop{op}{X}{V}{\tmcont x M}{\tmexccont N e {E_\op}} : \tyuser{X}{\Ueff}
  }

  \coopinfer{TyKernel-Op}{
    \Keff \equiv (\sig, E, S, C) \\
    \op \in \sig \\
    \Gamma \types V : A_\op \\\\
    \Gamma, x \of B_\op \types K : \tykernel{X}{\Keff} \\
    \big(
      \Gamma \vdash L_e : \tykernel{X}{\Keff}
    \big)_{e \in E_\op}
  }{
    \Gamma \types \tmop{op}{X}{V}{\tmcont x K}{\tmexccont L e {E_\op}} : \tykernel{X}{\Keff}
  }

  \coopinfer{TyUser-Kernel}{
    F \equiv
    \{ \tmreturn{x} \at c \mapsto N,
       (\tmraise{e} \at c \mapsto N_e)_{e \in E},
       (\tmkill{s} \mapsto N_s)_{s \in S}
    \}
    \\\\
    \Gamma \types K : \tykernel{X}{(\sig, E, S, C)} \\
    \Gamma \types W : C \\
    \Gamma, x \of X, c \of C \types N : \tyuser{Y}{(\sig, E')} \\
    \big(
      \Gamma, c \of C \types N_e : \tyuser{Y}{(\sig, E')}
    \big)_{e \in E} \\
    \big(
      \Gamma \types N_s : \tyuser{Y}{(\sig, E')}
    \big)_{s \in S} \\
  }{
    \Gamma \types \tmkernel{K}{W}{F} : \tyuser{Y}{(\sig, E')}
  }

  \coopinfer{TyKernel-User}{
   \Keff \equiv (\sig, E', S, C) \\
   \Gamma \types M : \tyuser{X}{(\sig, E)} \\\\
   \Gamma, x \of X \types K : \tykernel{Y}{\Keff} \\
   \big(
     \Gamma \types L_e : \tykernel{Y}{\Keff}
   \big)_{e \in E}
  }{
    \Gamma \types
    \tmuser{M}{
      \{ \tmreturn{x} \mapsto K,
         (\tmraise{e} \mapsto L_e)_{e \in E}
      \}
    }
    : \tykernel{Y}{\Keff}
  }
  \end{mathpar}
  \caption{Selected typing and subtyping rules.}
  \label{fig:typing-selected}
\end{figure}

Subtyping of ground types \textsc{Sub-Ground} is trivial, as it relates only equal types.
Subtyping of runners \textsc{Sub-Runner} and kernel computations
\textsc{Sub-Kernel} requires equality of the kernel state types~$C$ and~$C'$
because state is used invariantly in the kernel monad.
We leave it for future work to replace ${C \equiv C'}$ with a
\emph{lens}~\cite{Foster:Lenses} from~$C'$ to~$C$, i.e., maps $C' \to C$ and ${C' \times C \to C'}$
satisfying state equations analogous to \cref{ex:state}. It has been
observed~\cite{OConnor:Lens,Power:Comodels} that such a lens in fact amounts to
an ordinary runner for $C$-valued state.

The rules \textsc{TyUser-Op} and \textsc{TyKernel-Op} govern operation calls, where 
we have a success continuation which receives a value returned by a 
co-operation, and exceptional continuations which receive exceptions raised by co-operations.

The rule \textsc{TyUser-Run} requires that the runner $V$ implements \emph{all} the
operations $M$ can use, meaning that operations are \emph{not} implicitly propagated outside a $\tmkw{run}$ block (which is different from how handlers are sometimes implemented). Of course, the co-operations of the runner may call further external operations, as recorded by the signature~$\sig'$. Similarly, we require the finally block~$F$ to intercept all exceptions and signals that might be produced by the co-operations of $V$ or the user code $M$.
Such strict control is exercised throughout. For example, in 
\textsc{TyUser-Run}, \textsc{TyUser-Kernel}, and \textsc{TyKernel-User} we catch all the exceptions and signals that the code might produce.
One should judiciously relax these requirements in a language that is presented to
the programmer, and allow re-raising and re-sending clauses to be automatically inserted.



\subsection{Equational theory}
\label{sect:eqtheory}

We present $\lambdacoop$ as an \emph{equational calculus}, i.e., the interactions between
its components are described by equations. Such a presentation makes it easy to reason
about program equivalence.
There are three equality judgements
\begin{equation*}
\Gamma \types V \equiv W : X, 
\qquad
\Gamma \types M \equiv N : \tyuser{X}{\Ueff}, 
\qquad
\Gamma \types K \equiv L  : \tyuser{X}{\Keff}.
\end{equation*}
It is presupposed that we only compare well-typed expressions with the indicated types.
For the most part, the context and the type annotation on judgements will play no significant role,
and so we shall drop them whenever possible.

We comment on the computational equations for constructs characteristic
of~$\lambdacoop$, and refer the reader to the online appendix for other equations.
When read left-to-right, these equations explain the operational meaning of programs.

Of the three equations for $\tmkw{run}$, the first two specify that returned values and
raised exceptions are handled by the corresponding clauses,
\begin{align*}
  \tmrun{V}{W}{(\tmreturn{V'})}{F} &\equiv N[V'/x, W/c], 
  \\
  \tmrun{V}{W}{(\tmraise[X]{e})}{F} &\equiv N_{e}[W/c],
\end{align*}
where
$F \defeq \{\tmreturn{x} \at c \mapsto N,
   (\tmraise{e} \at c \mapsto N_e)_{e \in E},
   (\tmkill{s} \mapsto N_s)_{s \in S}
\}$.
The third equation below relates running an operation $\op$ with executing the corresponding co-operation~$K_\op$, 
where $R$ stands for the runner
$\tmrunner{(\tm{op}\,x \mapsto K_{\tm{op}})_{\tm{op} \in \sig}}{C}$:
\begin{multline*}
  \tmrun{R}{W}{(
    \tmop{op}{X}{V}{\tmcont x M}{\tmexccont {N'} {e'} {E_\op}}
    )}{F} \equiv {}
  \\
  \begin{aligned}[t]
     &\tmkernel{K_\op[V/x]}{W}{} \\
     &\qquad\big\{
         \begin{aligned}[t]
           &\tmreturn{x} \at c' \mapsto (\tmrun{R}{c'}{M}{F}), \\
           &\left(
               \tmraise{e'} \at c' \mapsto (\tmrun{R}{c'}{N'_{e'}}{F})
             \right)_{e' \in E_\op},\\
           &\left(
               \tmkill{s} \mapsto N_s
             \right)_{s \in S} \big\}
         \end{aligned}
 \end{aligned}
\end{multline*}
Because $K_\op$ is kernel code, it is executed in kernel mode, whose
$\tmkw{finally}$ clauses specify what happens afterwards: if $K_\op$ returns a value, or
raises an exception, execution continues with a suitable continuation, with~$R$
wrapped around it; and if $K_\op$ sends a signal, the corresponding finalisation code from $F$ is
evaluated.

The next bundle describes how kernel code is executed within user code:
\begin{align*}
  \tmkernel{(\tmreturn[C]{V})}{W}{F} &\equiv N[V/x, W/c], \\
  \tmkernel{(\tmraise[X \at C]{e})}{W}{F} &\equiv N_{e}[W/c], \\
  \tmkernel{(\tmkill[X \at C]{s})}{W}{F} &\equiv N_{s}, \\
  \tmkernel{(\tmgetenv[C]{\tmcont c K})}{W}{F} &\equiv \tmkernel{K[W/c]}{W}{F}, \\
  \tmkernel{(\tmsetenv{V}{K})}{W}{F} &\equiv \tmkernel{K}{V}{F}.
\end{align*}
We also have an equation stating that an operation called in kernel mode propagates out to
user mode, with its continuations wrapped in kernel mode:
\begin{multline*}
  \tmkernel{\tmop{op}{X}{V}{\tmcont x K}{\tmexccont L {e'} E}}{W}{F}
  \equiv {} \\
  \tmop{op}{X}{V}{\tmcont x {\tmkernel{K}{W}{F}}}{
    \left(
      \tmkernel{L_{e'}}{W}{F}
    \right)_{e' \in E}}.
\end{multline*}
Similar equations govern execution of user computations in kernel mode.

The remaining equations include standard $\beta\eta$-equations for
exception handling~\cite{Benton:ExceptionalSyntax}, deconstruction of products and sums,
algebraicity equations for operations~\cite{Pretnar:Thesis}, and the equations of kernel theory from \cref{sec:user-kernel-monads}, describing how $\tmkw{getenv}$ and $\tmkw{setenv}$ work, and how they interact with signals and other operations.



\section{Denotational semantics}
\label{sec:denotational-semantics}

We provide a coherent denotational semantics for~$\lambdacoop$, and 
prove it sound with respect to the equational theory given
in \cref{sect:eqtheory}. Having eschewed all forms of recursion,
we may afford to work simply over the category of sets and functions, while noting that
there is no obstacle to incorporating recursion at all levels and switching to domain theory,
similarly to the treatment of effect handlers in~\cite{Bauer:EffectSystem}.

\subsection{Semantics of types}
\label{sec:semantics-types}

The meaning of terms is most naturally defined by structural induction on their typing
derivations, which however are not unique in~$\lambdacoop$ due to subsumption rules.
Thus we must worry about devising a \emph{coherent} semantics, i.e., one in which all derivations
of a judgement get the same meaning.
We follow prior work on the semantics of effect systems for
handlers~\cite{Bauer:EffectSystem}, and proceed by first giving a \emph{skeletal}
semantics of~$\lambdacoop$ in which derivations are manifestly unique because the effect
information is unrefined. We then use the skeletal semantics as the frame upon which rests
a refinement-style coherent semantics of the effectful types of~$\lambdacoop$.

The \emph{skeletal} types are like $\lambdacoop$'s types, but with all effect information
erased. 
In particular, the ground types $A$, and hence the kernel state types $C$, do not change
as they contain no effect information. The skeletal value types are
\begin{equation*}
  P, Q
   \bnfis A
   \mid \tyunit
   \mid \tyempty
   \mid \typrod{P}{Q}
   \mid \tysum{P}{Q}
   \mid \tyfunskel{P}{\tyuserskel{Q}}
   \mid \tyfunKskel{P}{\tykernelskel{Q}{C}}
   \mid \tyrunnerskel{C}.
\end{equation*}
The skeletal versions of the user and kernel types are $\tyuserskel{P}$ and
$\tykernelskel{P}{C}$, respectively. It is best to think of the skeletal types as ML-style
types which implicitly over-approximate effect information by ``any effect is possible'',
an idea which is mathematically expressed by their semantics, as explained below.

First of all, the semantics of ground types is straightforward. One only needs to provide
sets denoting the base types $\tybase$, after which the ground types receive the standard
set-theoretic meaning, as given in \cref{fig:semantics-ground-skeletal-types}.

Recall that $\Ops$, $\Sigs$, and $\Excs$ are the sets of all operations, signals, 
and exceptions, and that each $\op \in \Ops$ has a signature
$\op : \tysigop{A_\op}{B_\op}{E_\op}$.
Let us additionally assume that there is a distinguished operation
$\rip \in \Ops$ with signature $\rip : \tysigop{\One}{\Zero}{\Zero}$ (otherwise
we adjoin it to $\Ops$). It ensures that the denotations of skeletal
user and kernel types are \emph{pointed} sets, while operationally~$\rip$ indicates
a \emph{runtime error}.

Next, we define the \emph{skeletal user and kernel monads} as
\begin{align*}
  \UUskel X &\defeq \UU{\Ops}{\Excs} X= \Tree{\Ops}{X + \Excs}, \\
  \KKskel{C} X &\defeq \KK{C}{\Ops}{\Excs}{\Sigs} X = (C \expto \Tree{\Ops}{(X + \Excs) \times C + \Sigs}),
\end{align*}
and $\RunnerSkel{C}$ as the set of all \emph{skeletal runners $\R$ (with state $C$)}, 
which are families of co-operations
$
  \{\coop_\R : \skelsem{A_\op} \to \KK{C}{\Ops}{E_\op}{\Sigs} \skelsem{B_\op} \}_{\op \in \Ops}
$.
Note that $\KK{C}{\Ops}{E_\op}{\Sigs}$ is a coproduct~\cite{Hyland:CombiningEffects} of monads
$C \expto \Tree{\Ops}{{-} \times C + \Sigs}$ and~$\Exc{E_\op}$, and thus the skeletal
runners are the effectful runners for the former monad, so long as we read the effectful
signatures $\op : \tysigop{A_\op}{B_\op}{E_\op}$ as ordinary algebraic ones $\op : A_\op \leadsto B_\op + E_\op$.
While there is no semantic difference between the two readings, there is one of intention:
$\KK{C}{\Ops}{E_\op}{\Sigs} \skelsem{B_\op}$ is a kernel computation that (apart from
using state and sending signals) returns values of type $B_\op$ and raises
exceptions $E_\op$, whereas $C \expto \Tree{\Ops}{(\skelsem{B_\op} + E_\op) \times C + \Sigs}$ returns values of type
$B_\op + E_\op$ and raises no exceptions. We prefer the former, as it reflects our treatment of exceptions as a
control mechanism rather than exceptional values.

These ingredients suffice for the denotation of skeletal types as sets, as given in
\cref{fig:semantics-ground-skeletal-types}. The user and kernel skeletal types
are interpreted using the respective skeletal monads, and hence the two function types as Kleisli 
exponentials.

\begin{figure}[ht]
  \centering
  \small
  \textbf{Ground types}
  \begin{gather*}
    \skelsem{\tybase} \defeq \cdots
    \qquad\qquad
    \skelsem{\tyunit} \defeq \One
    \qquad\qquad
    \skelsem{\tyempty} \defeq \Zero
    \\
    \skelsem{\typrod{A}{B}} \defeq \skelsem{A} \times \skelsem{B}
    \qquad\qquad
    \skelsem{\tysum{A}{B}} \defeq \skelsem{A} + \skelsem{B}
  \end{gather*}
  \textbf{Skeletal types}
  \begin{gather*}
    \begin{aligned}
      \skelsem{\typrod{P}{Q}} &\defeq \skelsem{P} \times \skelsem{Q}
      &
      \qquad
      \skelsem{\tyfun{P}{\tyuser{Q}{}}} &\defeq \skelsem{P} \expto \skelsem{\tyuser{Q}{}}
      \\
      \skelsem{\tysum{P}{Q}} &\defeq \skelsem{P} + \skelsem{Q}
      &
      \skelsem{\tyfunKskel{P}{\tykernelskel{Q}{C}}} &\defeq \skelsem{P} \expto \skelsem{\tykernelskel{Q}{C}}
    \end{aligned}
    \\
    \skelsem{\tyrunnerskel{C}} \defeq \RunnerSkel{\skelsem{C}}
    \qquad\quad
    \skelsem{\tyuserskel{P}} \defeq \UUskel \skelsem{P}
    \qquad\quad
    \skelsem{\tykernelskel{P}{C}} \defeq \KKskel{\skelsem{C}} \skelsem{P}
    \\
    \skelsem{x_1 : P_1, \ldots, x_n : P_n} \defeq \skelsem{P_1} \times \cdots \times \skelsem{P_n}
  \end{gather*}
  \caption{Denotations of ground and skeletal types.}
  \label{fig:semantics-ground-skeletal-types}
\end{figure}

We proceed with the semantics of effectful types. The \emph{skeleton} of a value
type~$X$ is the skeletal type $\skeleton{X}$ obtained by removing all effect
information, and similarly for user and kernel types, see \cref{fig:skeletons-and-denotations}.
%
%
\begin{figure}[ht]
  \centering
  \small
  \textbf{Skeletons}
  \begin{gather*}
  \skeleton{A} \defeq A
  \qquad
    \skeleton{(\tyrunner{\sig}{\sig'}{S}{C})} \defeq \tyrunnerskel{C}
  \qquad
    \skeleton{(\typrod{X}{Y})} \defeq \typrod{\skeleton{X}}{\skeleton{Y}}
  \\
  \skeleton{(\tyfun{X}{\tyuser{Y}{\Ueff}})} \defeq \tyfunskel{\skeleton{X}}{\skeleton{(\tyuser{Y}{\Ueff})}}
  \qquad
    \skeleton{(\tysum{X}{Y})} \defeq \tysum{\skeleton{X}}{\skeleton{Y}}
  \\
  \skeleton{(\tyfunK{X}{\tykernel{Y}{\Keff}})} \defeq
                                                     \tyfunKskel{\skeleton{X}}{\skeleton{(\tykernel{Y}{\Keff})}}
  \qquad
    \skeleton{(\tyuser{X}{\Ueff})} \defeq \tyuserskel{\skeleton{X}}
  \\
  \skeleton{(x_1 : X_1, \ldots, x_n : X_n)} \defeq (x_1 : \skeleton{X}_1, \ldots, x_n : \skeleton{X}_n)
  \qquad
    \skeleton{(\tykernel{X}{(\Sigma, E, S, C)})} \defeq \tykernelskel{\skeleton{X}}{C}
  \end{gather*}
  \abovedisplayskip=4pt 
  \textbf{Denotations}
  \begin{gather*}
    \begin{aligned}
    \sem{A} &\defeq \skelsem{A}
    &
    \sem{\typrod{X}{Y}} &\defeq \sem{X} \times \sem{X}
    \\
    \sem{\tyrunner{\sig}{\sig'}{S}{C}} &\defeq \Runner{\sig}{\sig'}{S}{\sem{C}}
    &
    \qquad\qquad
    \sem{\tysum{X}{Y}} &\defeq \sem{X} + \sem{X}
    \end{aligned}
    \\
    \begin{aligned}
      \sem{\tyfun{X}{\tyuser{Y}{\Ueff}}} &\defeq
      (\sem{X}, \skelsem{\skeleton{X}}) \subexpto (\sem{\tyuser{Y}{\Ueff}}, \skelsem{\skeleton{(\tyuser{Y}{\Ueff})}})
    \\
      \sem{\tyfunK{X}{\tykernel{Y}{\Keff}}} &\defeq
      (\sem{X}, \skelsem{\skeleton{X}}) \subexpto (\sem{\tykernel{Y}{\Keff}}, \skelsem{\skeleton{(\tykernel{Y}{\Keff})}})
    \end{aligned}
    \\
    \begin{aligned}
      \sem{\tyuser{X}{(\sig,E)}} &\defeq \UU{\sig}{E} \sem{X}
      &
      \sem{\tykernel{X}{(\sig, E, S, C)}} &\defeq \KK{\skelsem{C}}{\sig}{E}{S} \sem{X}
    \end{aligned}
    \\
    \sem{x_1 : X_1, \ldots, x_n : X_n} \defeq \sem{X_1} \times \cdots \times \sem{X_n}
  \end{gather*}
  \caption{Skeletons and denotations of types.}
  \label{fig:skeletons-and-denotations}
\end{figure}
%
%
We interpret a value type~$X$ as a subset
$\sem{X} \subseteq \skelsem{\skeleton{X}}$ of the denotation of its skeleton,
and similarly for user and computation types. In other words, we treat the
effectful types as \emph{refinements} of their skeletons. For this,
we define the operation $(X_0, X_1) \subexpto (Y_0, Y_1)$, for any  
$X_0 \subseteq X_1$ and $Y_0 \subseteq Y_1$, as the set of maps $X_1 \to Y_1$
restricted to $X_0 \to Y_0$: 
\begin{equation*}
  (X_0, X_1) \subexpto (Y_0, Y_1) \defeq
  \{ f : X_1 \to Y_1 \mid \forall x \in X_0 \,.\, f(x) \in Y_0 \}.
\end{equation*}
Next, observe that the user and the kernel monads preserve subset inclusions, in
the sense that $\UU{\sig}{E} X \subseteq \UU{\sig'}{E'} X'$ and
$\KK{C}{\sig}{E}{S} X \subseteq \KK{C}{\sig'}{E'}{S'} X'$ if
$\sig \subseteq \sig'$, $E \subseteq E'$, $S \subseteq S'$, and
$X \subseteq X'$. In particular, we always have
$\UU{\sig}{E} X \subseteq \UUskel X$ and
$\KK{C}{\sig}{E}{S} X \subseteq \KKskel{C} X$.
Finally, let $\Runner{\sig}{\sig'}{S}{C} \subseteq \RunnerSkel{C}$ be the subset of those
runners~$\R$ whose co-operations for $\sig$ factor through $\KK{C}{\sig'}{E_\op}{S}$,
i.e.,
$
  \coop_\R : \skelsem{A_\op} \to \KK{C}{\sig'}{E_\op}{S} \skelsem{B_\op}
  \subseteq \KK{C}{\Ops}{E_\op}{\Sigs} \skelsem{B_\op}, 
$
for each $\op \in \sig$.

Semantics of effectful types is given in \cref{fig:skeletons-and-denotations}.
From a category-theoretic viewpoint, it assigns meaning in the category $\Sub(\Set)$
whose objects are subset inclusions $X_0 \subseteq X_1$ and morphisms from
$X_0 \subseteq X_1$ to $Y_0 \subseteq Y_1$ those maps $X_1 \to Y_1$ that restrict to
$X_0 \to Y_0$. The interpretations of products, sums, and function types are precisely the
corresponding category-theoretic notions $\times$, $+$, and $\subexpto$ in $\Sub(\Set)$.
Even better, the pairs of submonads $\UU{\sig}{E} \subseteq \UUskel$ and
$\KK{C}{\sig}{E}{S} \subseteq \KKskel{C}$ are the ``$\Sub(\Set)$-variants'' of the user
and kernel monads.
Such an abstract point of view drives the interpretation of terms, given
below, and it additionally suggests how our semantics can be set up on top of a
category other than~$\Set$. For example, if we replace $\Set$ with the category $\Cpo$ of
$\omega$-complete partial orders, we obtain the domain-theoretic semantics
of effect handlers from~\cite{Bauer:EffectSystem} that models recursion and operations whose
signatures contain arbitrary types.

\subsection{Semantics of values and computations}
\label{sec:semant-valu-comp}

To give semantics to $\lambdacoop$'s terms, we introduce \emph{skeletal
  typing} judgements
\begin{align*}
  &\Gamma \typesskel V : P,
& &\Gamma \typesskel M : \tyuserskel{P},
& &\Gamma \typesskel K : \tykernelskel{P}{C}, 
\end{align*}
which assign skeletal types to values and computations.
In these judgements, $\Gamma$ is a \emph{skeletal context} which assigns skeletal types to variables.

The rules for these judgements are
obtained from $\lambdacoop$'s typing rules, by \emph{excluding} subsumption rules and by relaxing restrictions on effects. For example, 
the skeletal versions of the rules \textsc{TyValue-Runner} and \textsc{TyKernel-Kill} are
\begin{equation*}
  \infer{
    \big(
      \Gamma, x \of A_\op \typesskel K_\op : \tykernelskel{B_\op}{C}
    \big)_{\op \in \sig}
  }{
    \Gamma \typesskel
    \tmrunner{(\tm{op}\,x \mapsto K_{\tm{op}})_{\tm{op} \in \Sigma}}{C} :
    \tyrunnerskel{C}
  }
  \qquad\qquad
  \infer{
     s \in \Sigs
  }{
     \Gamma \typesskel \tmkill[X \at C]{s} : \tykernelskel{\skeleton X}{C}
  }
\end{equation*}
The relationship between effectful and skeletal typing is summarised as follows:

\begin{proposition}
\label{prop:skeletaltypes}
(1) Skeletal typing derivations are unique.
(2) If $\subty X Y$, then $\skel X = \skel Y$, and analogously for subtyping of user and kernel types.
(3) If ${\Gamma \types V : X}$, then ${\skeleton{\Gamma} \typesskel V : \skeleton{X}}$, and
analogously for user and kernel computations.
\end{proposition}

\begin{proof}
  We prove (1) by induction on skeletal typing derivations, and
  (2) by induction on subtyping derivations. 
  For (1), we further use the occasional type annotations, and the 
  absence of skeletal subsumption rules.
  For proving (3), suppose that $\mathcal{D}$ is a derivation of
  $\Gamma \types V : X$. We
  may translate $\mathcal{D}$ to its \emph{skeleton} $\skeleton{\mathcal{D}}$ deriving
  $\skeleton{\Gamma} \typesskel V : \skeleton{X}$ by replacing typing rules
  with matching skeletal ones, skipping subsumption rules due to (2).
  Computations are treated similarly. \qed
\end{proof}

To ensure semantic coherence, we first define the \emph{skeletal semantics} of skeletal
typing judgements, 
$\skelsem{\Gamma \typesskel V : P} : \skelsem{\Gamma} \to \skelsem{P}$,
$\skelsem{\Gamma \typesskel M : \tyuserskel{P}} : \skelsem{\Gamma} \to \skelsem{\tyuserskel{P}}$,
and
$\skelsem{\Gamma \typesskel K : \tykernelskel{P}{C}} : \skelsem{\Gamma} \to \skelsem{\tykernelskel{P}{C}}$, 
by induction on their (unique) derivations.

Provided maps $\skelsem{A_1} \times \cdots \times \skelsem{A_n} \to \skelsem{B}$
denoting ground constants $\tmconst{f}$, 
values are interpreted in a standard
way, using the bi-cartesian closed structure of sets, except for a runner
$\tmrunner{(\op\, x \mapsto K_\op)_{\op \in \sig}}{C}$, which is interpreted at
an environment $\gamma \in \skelsem{\Gamma}$ as the skeletal runner
$\{\coop : \skelsem{A_\op} \to \KK{\skelsem{C}}{\Ops}{E_\op}{\Sigs} \skelsem{B_\op}\}_{\op
  \in \Ops}$, given by
\begin{equation*}
  \coop\,a \defeq
  (\cond
        {\op \in \sig}
        {\rho(\skelsem{\Gamma, x : A_\op \typesskel K_\op : \tykernelskel{B_\op}{C}}(\gamma, a))}
        {\rip}).
\end{equation*}
Here the map 
$\rho : \KKskel{\skelsem{C}} \skelsem{B_\op} \to \KK{\skelsem{C}}{\Ops}{E_\op}{\Sigs}
\skelsem{B_\op}$ is the skeletal kernel theory homomorphism characterised by the equations
\begin{gather*}
  \rho (\retTree \, b) = \retTree \, b, 
  \qquad
  \rho (\op'(a', \kappa, (\nu_e)_{e \in E_{\op'}})) =
     \op'(a', \rho \circ \kappa, (\rho(\nu_e))_{e \in E_{\op'}}), 
  \\
  \begin{aligned}
    \rho (\siggetenv\,\kappa) &= \siggetenv (\rho \circ \kappa), 
    &
    \rho (\sigraise\,e) &= (\cond {e \in E_\op} {\sigraise\, e} {\rip}), 
    \\
    \rho (\sigsetenv (c, \kappa)) &= \siggetenv (c, \rho \circ \kappa), 
    &
    \rho (\sigkill\,s) &= \sigkill\,s.
  \end{aligned}
\end{gather*}
The purpose of $\rip$ in the definition of $\coop$ is to model a runtime error
when the runner is asked to handle an unexpected operation, while~$\rho$ makes sure
that~$\coop$ raises at most the exceptions~$E_\op$, as prescribed by the 
signature of $\op$.

User and kernel computations are interpreted as elements of the corresponding skeletal
user and kernel monads. Again, most constructs are interpreted in a standard way:
$\tmkw{return}$s as the units of the monads; the operations $\tmkw{raise}$,
$\tmkw{kill}$, $\tmkw{getenv}$, $\tmkw{setenv}$, and $\tm{op}$s as the corresponding
algebraic operations; and $\tmkw{match}$ statements as the corresponding semantic
elimination forms. The interpretation of exception handling offers no surprises, e.g., as
in~\cite{Plotkin:HandlingEffects}, as long as we follow the strategy of treating
unexpected situations with the runtime error~$\rip$.

The most interesting part of the interpretation is the semantics of
\begin{align}
  \label{eq:semantics-run}
  & {\Gamma \typesskel (\tmrun{V}{W}{M}{F}) : \tyuserskel{Q}}, 
\end{align}
where $ F \defeq \{ \tmreturn{x} \at c \mapsto N, 
             (\tmraise{e} \at c \mapsto N_e)_{e \in E},
             (\tmkill{s} \mapsto N_s)_{s \in S} \}$.
At an environment $\gamma \in \skelsem{\Gamma}$, $V$ is interpreted as a skeletal runner
with state $\skelsem{C}$, which induces a monad morphism
$\runh : \Tree{\Ops}{-} \to (\skelsem{C} \expto \Tree{\Ops}{{-} \times \skelsem{C} + \Sigs})$, 
as in the proof of \cref{prop:monadmorphism}. 
Let
$f : \KKskel{\skelsem{C}} \skelsem{P} \to (\skelsem{C} \expto \UUskel \skelsem{Q})$ be the skeletal
kernel theory homomorphism characterised by the equations
\begin{gather}
  \label{eq:finally-map}
  \begin{aligned}
    f(\retTree{p}) &= \lam{c} \skelsem{\Gamma, x \of P, c \of C \typesskel N : Q}(\gamma, p, c), 
    \\
    f(\op(a, \kappa, (\nu_e)_{e \in E_\op})) &= \lam{c} \op(a, \lam{b} f (\kappa \, b) \, c, (f(\nu_e)\, c)_{e \in E_\op}), 
    \\
    f(\sigraise\,e) &= \lam{c}
      (\cond
        {e \in E}
        {\skelsem{\Gamma, c : C \typesskel N_e : Q}(\gamma, c)}
        {\rip}), 
    \\
    f(\sigkill\,s) &= \lam{c}
      (\cond
        {s \in S}
        {\skelsem{\Gamma \typesskel N_s : Q}\,\gamma}
        {\rip}), 
  \end{aligned} 
  \\
 \notag
  f(\siggetenv\,\kappa) = \lam{c} f (\kappa \, c) \, c, 
  \qquad\qquad
  f(\sigsetenv(c', \kappa)) = \lam{c} f \, \kappa \, c'.
\end{gather}
The interpretation of~\eqref{eq:semantics-run} at $\gamma$ is
$
  f(\runh_{\skelsem{P} + \Excs} (\skelsem{\Gamma \typesskel M : \tyuserskel{P}} \, \gamma))
  \, (\skelsem{\Gamma \typesskel W : C} \, \gamma)
$,
which reads: map the interpretation of $M$ at~$\gamma$ from the skeletal user
monad to the skeletal kernel monad using~$\runh$ (which models the operations of $M$ by the
cooperations of~$V$), and from there using~$f$ to a map
$\skelsem{C} \expto \UUskel \skelsem{Q}$, that is then applied to the initial kernel state, namely, the
interpretation of $W$ at~$\gamma$.

We interpret the context switch $\Gamma \typesskel \tmkernel{K}{W}{F} : \tyuserskel{Q}$
at an environment $\gamma \in \skelsem{\Gamma}$
as $f(\skelsem{\Gamma \typesskel K : \tykernelskel{P}{C}} \, \gamma)\,
  (\skelsem{\Gamma \typesskel W : C} \, \gamma)$, where $f$ is the map~\eqref{eq:finally-map}. 
Finally, $\tmkw{user}$ context switch is interpreted much like exception handling.

We now define coherent semantics of $\lambdacoop$'s typing derivations by passing
through the skeletal semantics. Given a derivation $\mathcal{D}$ of $\Gamma \types V : X$,
its skeleton $\skeleton{\mathcal{D}}$ derives
$\skeleton{\Gamma} \typesskel V : \skeleton{X}$.
We identify the denotation of~$V$ with the skeletal one,
\begin{equation*}
  \sem{\Gamma \types V : X} \defeq \skelsem{\skeleton{\Gamma} \typesskel V : \skeleton{X}} :
  \skelsem{\skeleton{\Gamma}} \to \skelsem{\skeleton{X}}.
\end{equation*}
All that remains is to check that $\sem{\Gamma \types V : X}$ restricts to
$\sem{\Gamma} \to \sem{X}$. This is accomplished by induction on~$\mathcal{D}$. 
The only interesting step is subsumption, which relies on a further
observation that $\subty{X}{Y}$ implies $\sem{X} \subseteq \sem{Y}$. Typing derivations for user and kernel
computations are treated analogously.

\subsection{Coherence, soundness, and finalisation theorems}
\label{sec:finalisation-theorem}

We are now ready to prove a theorem that guarantees execution of finalisation code. But
first, let us record the fact that the semantics is coherent and sound.

\begin{theorem}[Coherence and soundness]
  \label{thm:soundness-coherence}%
  The denotational semantics of $\lambdacoop$ is coherent, and it is sound for the equational
  theory of~$\lambdacoop$ from \cref{sect:eqtheory}.
\end{theorem}

\begin{proof}
  Coherence is established by construction: any two derivations of the same
  typing judgement have the same denotation because they are both
  (the same) restriction of skeletal semantics.
  For proving soundness, one just needs to unfold
  the denotations of the left- and right-hand sides of equations from \cref{sect:eqtheory},
  and compare them, where some cases rely on suitable substitution lemmas. \qed
\end{proof}

To set the stage for the finalisation theorem, let us 
consider the 
computation $\tmrun{V}{W}{M}{F}$, well-typed by the rule \textsc{TyUser-Run} from
\cref{fig:typing-selected}. At an environment $\gamma \in \sem{\Gamma}$, the finalisation
clauses~$F$ are captured semantically by the \emph{finalisation map}
$\phi_\gamma : (\sem{X} + E) \times \sem{C} + S \to \sem{\tyuser{Y}{(\sig',E')}}$, given by
\begin{align*}
    \phi_\gamma(\iota_1(\iota_1\, x, c)) &\defeq \sem{\Gamma, x \of X, c \of C \types N : \tyuser{Y}{(\sig',E')}}(\gamma, x, c),  \\
    \phi_\gamma(\iota_1(\iota_2\, e, c)) &\defeq \sem{\Gamma, c \of C \types N_e : \tyuser{Y}{(\sig',E')}}(\gamma, c),  \\
    \phi_\gamma(\iota_2(s)) &\defeq \sem{\Gamma \types N_s : \tyuser{Y}{(\sig',E')}} \, \gamma.
\end{align*}
With~$\phi$ in hand, we may formulate the finalisation theorem for $\lambdacoop$,
stating that the semantics of~$\tmrun{V}{W}{M}{F}$ is a computation tree all of
whose branches end with finalisation clauses from~$F$. Thus, unless some enveloping
runner sends a signal, finalisation with $F$ is guaranteed to take place.

\begin{theorem}[Finalisation]
  \label{thm:finalisation}%
  A well-typed $\tmkw{run}$ factors through finalisation:
  \begin{equation*}
    \sem{\Gamma \types (\tmrun{V}{W}{M}{F}) : \tyuser{Y}{(\sig',E')}}\,\gamma = \lift{\phi_\gamma}\,t, 
  \end{equation*}
  for some $t \in \Tree{\sig'}{(\sem{X} + E) \times \sem{C} + S}$.
\end{theorem}

\begin{proof}
  We first prove that $f \, u \, c = \lift{\phi_\gamma}(u\,c)$ holds for all
  $u \in \KK{\sem{C}}{\sig'}{E}{S} \sem{X}$ and $c \in \sem{C}$, where $f$ is the
  map~\eqref{eq:finally-map}. The proof proceeds by computational induction
  on~$u$~\cite{Plotkin:Logic}. The finalisation statement is then just the special case with
  $u \defeq \runh_{\sem{X} + E} (\sem{\Gamma \types M : \tyuser{X}{(\sig,E)}} \, \gamma)$ and
  $c \defeq \sem{\Gamma \types W : C} \, \gamma$. \qed
\end{proof}



\section{Runners in action}
\label{sect:examples}

Let us show examples that demonstrate how runners can be usefully combined to provide flexible
resource management.
We implemented these and other examples in the language \pl{Coop} and a library
\pl{Haskell-Coop}, see \cref{sect:implementation}.

To make the code more understandable, we do not adhere strictly to the syntax of
$\lambdacoop$, e.g., we use the generic versions of effects~\cite{Plotkin:AlgOperations}, 
as is customary in programming, and effectful initialisation of 
kernel state as discussed in \cref{sec:runn-as-progr}.

\begin{example}[Nesting]
\label{ex:nesting}
  In \cref{ex:file-IO}, we considered a runner $\mathsf{fileIO}$ for basic file
  operations. Let us suppose that $\mathsf{fileIO}$ is implemented by immediate calls to the
  operating system. Sometimes, we might prefer to accumulate writes and
  commit them all at once, which can be accomplished by interposing between $\mathsf{fileIO}$
  and user code the following runner $\mathsf{accIO}$, which accumulates writes in its state:
\begin{lstlisting}
{ write s' -> let s = getenv () in setenv (concat s s') }$_\tm{string}$
\end{lstlisting}
By \emph{nesting} the runners, and calling the outer $\tm{write}$ (the one of $\mathsf{fileIO}$) only in the finalisation
code for $\mathsf{accIO}$, the accumulated writes are commited all at once:
\begin{lstlisting}
using fileIO @ (open "hello.txt") run
  using accIO @ (return "") run
    write "Hello, world."; write "Hello, again."
  finally { return x @ s -> write s; return x }
finally { return x @ fh -> ... , raise QuotaExceeded @ fh -> ... , kill IOError -> ... }
\end{lstlisting}
\end{example}

\begin{example}[Instrumentation]
Above, $\mathsf{accIO}$ implements the same signature as $\mathsf{fileIO}$ and
thus intercepts operations without the user code being aware of it.
This kind of invisibility
can be more generally used to implement \emph{instrumentation}:
\begin{lstlisting}
using { ..., op x -> let c = getenv () in setenv (c+1); op x, ... }$_\tm{int}$ @ (return 0) run
  $M$
finally { return x @ c -> report_cost c; return x, ... }
\end{lstlisting}
Here the interposed runner implements all operations of some enveloping runner,
by simply forwarding them, 
while also measuring computational cost 
by counting the total number of operation calls, 
which is then reported during finalisation.
\end{example}

\begin{example}[ML-style references]
  \label{ex:ml-ref}
  Continuing with the theme of nested runners, they can also be used to implement abstract
  and safe interfaces to low-level resources. For instance, suppose we have a low-level
  implementation of a memory heap that potentially allows unsafe memory access, and we would
  like to implement ML-style references on top of it. A good first attempt is the runner
\begin{lstlisting}
{ ref x -> let h = getenv () in
              let (r,h') = malloc h x in
              setenv h'; return r,
  get r -> let h = getenv () in memread h r,
  put (r, x) -> let h = getenv () in memset h r x }$_\tm{heap}$
\end{lstlisting}
  which has the desired interface, but still suffers from three deficiencies that can be
  addressed with further language support. First, \emph{abstract types} would let us hide the
  fact that references are just memory locations, so that the user code could never devise
  invalid references or otherwise misuse them. Second, our simple typing discipline forces
  all references to hold the same type, but in reality we want them to have different types.
  This could be achieved through quantification over types in the low-level implementation of 
  the heap, as we have done in the \pl{Haskell-Coop} library using \pl{Haskell}'s $\mathsf{forall}$.
  Third, user code could hijack a reference and
  misuse it out of the scope of the runner, which is difficult to prevent. In practice the
  problem does not occur because, so to speak, the runner for references is at the very top level,
  from which user code cannot escape.
\end{example}

\begin{example}[Monotonic state]
Nested runners can also implement access restrictions to resources,
with applications in security~\cite{DelignatLavaud:TLS}. For example, we can restrict
the references from the previous example to be used \emph{monotonically} by associating
a preorder with each reference, which assignments then have to obey.
This idea is similar to how monotonic state is implemented in the $\mathrm{F}^{*}$
language~\cite{Ahman:RecallingAWitness}, except that we make dynamic checks where
$\mathrm{F}^{*}$ statically uses dependent types.

While we could simply modify the previous example, it is better to implement a new runner which
is nested inside the previous one, so that we obtain a modular solution that works with
\emph{any} runner implementing operations $\tm{ref}$, $\tm{get}$, and $\tm{put}$:
\begin{lstlisting}
{ mref x rel -> let r = ref x in
                    let m = getenv () in
                    setenv (add m (r,rel)); return r,
  mget r -> get r,
  mput (r, y) -> let x = get r in
                     $\,$let m = getenv () in
                     $\,$match (sel m r) with
                     $\,$| inl rel -> if (rel x y) then put (r, y)
                                                 $\,\,$else raise MonotonicityViolation
                     $\,$| inr () -> kill NoPreoderFound }$_{\tm{map} (\tm{ref}, \tm{intRel})}$
\end{lstlisting}
The runner's state is a map from references to preorders on integers. The co-operation $\tm{mref~x~rel}$ creates a new
reference $\tm{r}$ initialised with~$\tm{x}$ (by calling $\tm{ref}$ of the outer runner), and 
then adds the pair $\tm{(r,rel)}$ to the map stored in the runner's state. Reading is delegated to the outer runner, while assignment 
first checks that the new state is larger than the old one, according to the associated preorder. If the
preorder is respected, the runner proceeds with assignment (again delegated to the outer runner), otherwise it reports a
monotonicity violation. We may not assume that every reference has an associated preorder,
because user code could pass to $\tm{mput}$ a reference that was created earlier outside 
the scope of the runner. If this happens, the runner simply kills the offending user code with
a signal.
\end{example}

\begin{example}[Pairing]
\label{ex:pairing}
  Another form of modularity is achieved by \emph{pairing} runners. Given two runners
  $\tmrunner{(\tm{op}\,x \mapsto K_{\tm{op}})_{\tm{op} \in \sig_1}}{C_1}$ and
  $\tmrunner{(\tm{op'}\,x \mapsto K_{\tm{op'}})_{\tm{op'} \in \sig_2}}{C_2}$, e.g., for state
  and file operations, we can use them side-by-side by combining them into a single runner
  with operations $\sig_1 + \sig_2$ and kernel state $C_1 \times C_2$, as follows (the
  co-operations $\tm{op}'$ of the second runner are treated symmetrically):
\begin{lstlisting}
{ op x -> let (c,c') = getenv () in
             user
               kernel ($K_\op$ x) @ c finally {
                  return y @ c'' -> return (inl (inl y, c'')),
                  (raise e @ c'' -> return (inl (inr e, c'')))$_{e \in E_\op}$, 
                  (kill s -> return (inr s))$_{s \in S_1}$}
             with {
               return (inl (inl y, c'')) -> setenv (c'', c'); return y,
               return (inl (inr e, c'')) -> setenv (c'', c'); raise e,
               return (inr s) -> kill s},
  op' x -> ... , ... }$_{C_1 \times C_2}$
\end{lstlisting}
Notice how the inner $\tmkw{kernel}$ context switch passes to the co-operation~$K_\op$ 
only its part of the combined state, and
how it returns the result of~$K_\op$ in a reified form (which requires treating exceptions and signals 
as values). The outer $\tmkw{user}$ context switch then receives this reified result, updates the combined state,
and forwards the result (return value, exception, or signal) in unreified form.
\end{example}



\section{Implementation}
\label{sect:implementation}

We accompany the theoretical development with two implementations of~$\lambdacoop$: a
prototype language \pl{Coop}~\cite{bauer19:Coop}, and a \pl{Haskell} library
\pl{Haskell-Coop}~\cite{ahman19:HaskellCoop}.

\pl{Coop}, implemented in \pl{OCaml}, demonstrates what a more fully-featured
language based on $\lambdacoop$ might look like.
It implements a bi-directional variant of $\lambdacoop$'s type system, extended
with type definitions and algebraic datatypes, to provide algorithmic typechecking and
type inference. The operational semantics is based on the computation rules of the
equational theory from \cref{sect:eqtheory}, but extended with general
recursion, pairing of runners from \cref{ex:pairing}, and an interface to the 
\pl{OCaml} runtime called \emph{containers}---these are essentially top-level runners
defined directly in \pl{OCaml}. They are a modular and systematic way of
offering several possible top-level runtime environments to the programmer.

The \pl{Haskell-Coop} library is a shallow embedding of $\lambdacoop$ in 
\pl{Haskell}. The implementation closely follows the denotational
semantics of~$\lambdacoop$. For instance, user and kernel monads are
implemented as corresponding \pl{Haskell} monads. Internally, the library
uses the \pl{Freer} monad of Kiselyov~\cite{Kiselyov:Freer}
to implement free model monads for given signatures of operations.
The library also provides a means to run user code via \pl{Haskell}'s top-level monads.
For instance, code that performs input-output operations may be run in \pl{Haskell}'s $\tm{IO}$ monad.

\pl{Haskell}'s advanced features
make it possible to use \pl{Haskell-Coop} to implement several 
extensions to examples from \cref{sect:examples}.
For instance, we implement ML-style state that allow references
holding arbitrary values (of different types), and state 
that uses \pl{Haskell}'s type system to track
which references are alive.
The library also provides pairing of runners from \cref{ex:pairing}, e.g., 
to combine state and input-output.
We also use the library to demonstrate that
\emph{ambient functions} from the \pl{Koka} language~\cite{Leijen:Ambients}
can be implemented with runners by treating their binding
and application as co-operations.
(These are functions that are bound dynamically but evaluated
in the lexical scope of their binding.)



\section{Related work}
\label{sect:relatedwork}

Comodels and (ordinary) runners have been used as a natural
model of stateful top-level behaviour.
For instance, Plotkin and Power~\cite{Plotkin:TensorsOfModels} have given a treatment of operational 
semantics using the tensor product of a model and a comodel.
Recently, Katsumata, Rivas, and Uustalu have generalised this interaction of models and comodels
to monads and comonads~\cite{Katsumata:InteractionLaws}.
An early version of \pl{Eff}~\cite{Bauer:AlgebraicEffects} implemented \emph{resources},
which were a kind of stateful runners, although they lacked satisfactory theory.
Uustalu~\cite{Uustalu:Runners} has pointed out that runners are the additional
structure that one has to impose on state
to run algebraic effects statefully.
Møgelberg and Staton's~\cite{Mogelberg:LinearUsageOfState} linear-use state-passing
translation also relies on equipping the state with a comodel
structure for the effects at hand. Our runners arise
when their setup is specialised to a certain Kleisli adjunction.

Our use of kernel state is analogous to the use
of parameters in parameter-passing 
handlers~\cite{Plotkin:HandlingEffects}: their 
$\tmkw{return}$ clause also provides a form of finalisation, as the 
final value of the parameter is available.
There is however no guarantee of
finalisation happening because handlers need not use the
continuation linearly.

The need to tame the excessive generality of handlers, and willingness to
give it up in exchange for efficiency and predictability, has recently
been recognised by \pl{Multicore OCaml}'s implementors, who
have observed that in practice most handlers resume 
continuations precisely once~\cite{Dolan:MulticoreOCaml}. In exchange for impressive efficiency, they require
continuations to be used linearly by default, whereas discarding and copying
must be done explicitly, incurring additional cost.
Leijen~\cite{Leijen:Finalisation} has extended 
handlers in \pl{Koka} with a $\tmkw{finally}$ clause, whose 
semantics ensures that finalisation happens whenever a handler
discards its continuation.
Leijen also added an $\tmkw{initially}$ clause to
parameter-passing handlers, which is used to compute the initial value of the parameter before handling, but that 
gets executed again every time the handler resumes its continuation.



\section{Conclusion and future work}
\label{sect:conclusion}

We have shown that effectful runners form a mathematically natural
and modular model of resources, modelling not only the top level external resources,
but allowing programmers to also define their own intermediate ``virtual machines''.
Effectful runners give rise to a bona fide programming concept, an idea we have captured
in a small calculus, called $\lambdacoop$, which we have implemented both as a
language and a library. We have given $\lambdacoop$
an algebraically natural denotational semantics, and shown how to program
with runners through various examples.

We leave combining runners and general effect handlers for future work.
As runners are essentially affine handlers, inspired by \pl{Multicore OCaml}
we also plan to investigate efficient compilation for runners.
On the theoretical side, by developing semantics in a $\Sub(\Cpo)$-enriched
setting~\cite{Power:EnrichedLawvereTheories}, we plan to support recursion at
all levels, and remove the distinction between ground and arbitrary types.
Finally, by using proof-relevant subtyping~\cite{Saleh:ExplicitSybTyping} and
synthesis of lenses~\cite{Miltner:SynthesizingSymmLenses}, we plan to upgrade 
subtyping from a simple inclusion to relating types by lenses.


\paragraph{Acknowledgements}
We thank Daan Leijen for useful discussions about initialisation and finalisation in \pl{Koka}, 
as well as ambient values and ambient functions. We thank  
Guillaume Munch-Maccagnoni and Matija Pretnar for discussing  
resources and potential future directions for $\lambdacoop$. We are also grateful to the participants 
of the NII Shonan Meeting ``Programming and reasoning with algebraic effects 
and effect handlers'' for feedback on an early version of this work.

\vspace{-10pt}


\begin{flushleft}
\begin{tabular}{@{}p{0.85\textwidth} l@{}}
This project has received funding from the European Union’s Horizon 2020 research and innovation programme under the Marie Sk\l{}odowska-Curie grant agreement No 834146.
&
\raisebox{-0.8cm}{
\includegraphics[width=1.5cm]{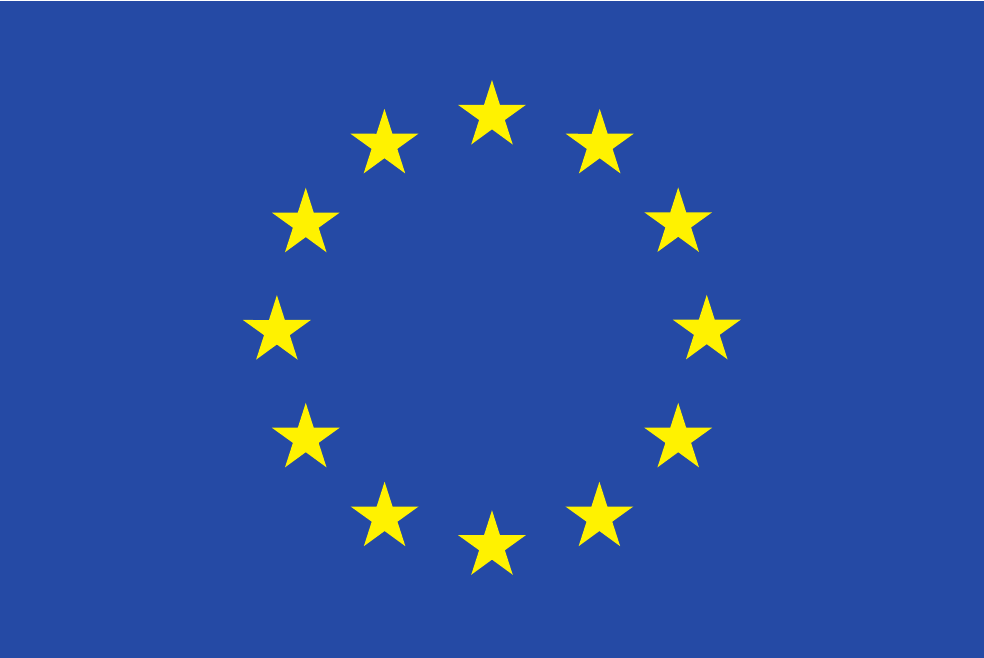}
}
\\
\multicolumn{2}{p{\textwidth}}{
\quad
This material is based upon work supported by the Air Force Office of Scientific Research under award number FA9550-17-1-0326.
}
\end{tabular}
\end{flushleft}


\bibliography{references}

\vfill

{\small\medskip\noindent{\bf Open Access} This chapter is licensed under the terms of the Creative Commons\break Attribution 4.0 International License (\url{http://creativecommons.org/licenses/by/4.0/}), which permits use, sharing, adaptation, distribution and reproduction in any medium or format, as long as you give appropriate credit to the original author(s) and the source, provide a link to the Creative Commons license and indicate if changes were made.}

{\small \spaceskip .28em plus .1em minus .1em The images or other third party material in this chapter are included in the chapter's Creative Commons license, unless indicated otherwise in a credit line to the material.~If material is not included in the chapter's Creative Commons license and your intended\break use is not permitted by statutory regulation or exceeds the permitted use, you will need to obtain permission directly from the copyright holder.}

\medskip\noindent\includegraphics{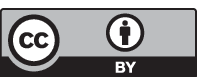}

\appendix
\makeatletter\def\@seccntformat#1{\appendixname~\csname the#1\endcsname: }\makeatother

\section{Typing rules of $\lambdacoop$}
\label{sec:typing-rules}

In this appendix we give the complete typing rules for $\lambdacoop$. We refer to
\cref{fig:lambdacoop-types,fig:lambdacoop-terms} for the syntax of types, values, 
and user and kernel computations.
For each operation symbol $\op \in \Ops$, we assume a given and fixed operation signature
\begin{equation*}
  \op : \tysigop{A_\op}{B_\op}{E_\op},
\end{equation*}
and for each ground constant $\tmconst{f}$, we assume 
a signature $\tmconst{f} : (A_1,\ldots,A_n) \to B$, 
both of which the typing rules refer to without further ado.
Values, and user and kernel computations each have a typing and a subtyping judgement of the
form
\begin{align*}
  &\Gamma \types V : X,
& &\Gamma \types M : \tyuser{X}{\Ueff},
& &\Gamma \types K : \tykernel{X}{\Keff},\\
  &X \sub Y,
& &\tyuser{X}{\Ueff} \sub \tyuser{Y}{\Veff},
& &\tykernel{X}{\Keff} \sub \tykernel{Y}{\Leff}.
\end{align*}
where $\Gamma$ is the customary typing context assigning value types to variables. 
The subtyping rules are given in \cref{fig:subtyping}, and the 
typing rules in \cref{fig:typing-values,fig:typing-user,fig:typing-kernel}.

\begin{figure}[p]
  \small
  \begin{mathpar}
    \coopinfer{Sub-Ground}{ }{A \sub A}

    \coopinfer{Sub-Product}{
      X \sub X' \\
      Y \sub Y'
    }{
      X \times Y \sub X' \times Y'
    }

    \coopinfer{Sub-Sum}{
      X \sub X' \\
      Y \sub Y'
    }{
      X + Y \sub X' + Y'
    }

    \coopinfer{Sub-UserFun}{
      X' \sub X \\
      \tyuser{Y}{\Ueff} \sub \tyuser{Y'}{\Ueff'}
    }{
      \tyfun{X}{\tyuser{Y}{\Ueff}} \sub \tyfun{X'}{\tyuser{Y'}{\Ueff'}}
    }

    \coopinfer{Sub-KernelFun}{
      X' \sub X \\
      \tykernel{Y}{\Keff} \sub \tykernel{Y'}{\Keff'}
    }{
      \tyfunK{X}{\tykernel{Y}{\Keff}} \sub \tyfunK{X'}{\tykernel{Y'}{\Keff'}}
    }

    \coopinfer{Sub-Runner}{
      \sig_1' \subseteq \sig_1 \\
      \sig_2 \subseteq \sig_2' \\
      S \subseteq S' \\
      C \equiv C'
    }{
      \tyrunner{\sig_1}{\sig_2}{S}{C} \sub \tyrunner{\sig_1'}{\sig_2'}{S'}{C'}
    }

    \coopinfer{Sub-User}{
      X \sub X' \\
      \sig \subseteq \sig' \\
      E \subseteq E'
    }{
      \tyuser{X}{(\sig, E)} \sub \tyuser{X'}{(\sig', E')}
    }

    \coopinfer{Sub-Kernel}{
      X \sub X' \\
      \sig \subseteq \sig' \\
      E \subseteq E' \\\\
      S \subseteq S' \\
      C \equiv C'
    }{
      \tykernel{X}{(\sig, E, S, C)} \sub \tykernel{X'}{(\sig', E', S', C')}
    }

    \coopinfer{Subsume-Value}{
      \Gamma \types V : X \\
      X \sub X'
    }{
      \Gamma \types V : X'
    }

    \coopinfer{Subsume-User}{
      \Gamma \types M : \tyuser{X}{\Ueff} \\
      \tyuser{X}{\Ueff} \sub \tyuser{X'}{\Ueff'}
    }{
      \Gamma \types M : \tyuser{X'}{\Ueff'}
    }

    \coopinfer{Subsume-Kernel}{
      \Gamma \types K : \tykernel{X}{\Keff} \\
      \tykernel{X}{\Keff} \sub \tykernel{X'}{\Keff'}
    }{
      \Gamma \types M : \tykernel{X'}{\Keff'}
    }
  \end{mathpar}
  \caption{Subtyping and subsumption rules.}
  \label{fig:subtyping}
\end{figure}

\begin{figure}[p]
  \centering
  \small
\begin{mathpar}

  \coopinfer{TyValue-Var}{
    \Gamma(x) \equiv X
  }{
    \Gamma \types x : X
  }
  
  \coopinfer{TyValue-Const}{
    (\Gamma \types V_i : A_i)_{1 \leq i \leq n}
  }{
    \Gamma \types \tmconst{f}(V_1, \ldots , V_n) : B
  }

  \coopinfer{TyValue-Unit}{
  }{
    \Gamma \types \tmunit : \tyunit
  }

  \coopinfer{TyValue-Pair}{
    \Gamma \types V : X \\
    \Gamma \types W : Y
  }{
    \Gamma \types \tmpair{V}{W} : \typrod{X}{Y}
  }

  \coopinfer{TyValue-Inl}{
    \Gamma \types V : X
  }{
    \Gamma \types \tminl[X,Y]{V} : X + Y
  }

  \coopinfer{TyValue-Inr}{
    \Gamma \types W : Y
  }{
    \Gamma \types \tminr[X,Y]{W} : X + Y
  }

  \coopinfer{TyValue-UserFun}{
    \Gamma, x \of X \types M : \tyuser{Y}{\Ueff}
  }{
    \Gamma \types \tmfun{x : X}{M} : \tyfun{X}{\tyuser{Y}{\Ueff}}
  }

  \coopinfer{TyValue-KernelFun}{
    \Gamma, x \of X \types K : \tykernel{Y}{\Keff}
  }{
    \Gamma \types \tmfunK{x : X}{K} : \tyfunK{X}{\tykernel{Y}{\Keff}}
  }

  \coopinfer{TyValue-Runner}{
    \big(
      \Gamma, x \of A_\op \types K_\op : \tykernel{B_\op}{(\sig', E_\op, S, C)}
    \big)_{\op \in \sig}
  }{
    \Gamma \types
    \tmrunner{(\tm{op}\,x \mapsto K_{\tm{op}})_{\tm{op} \in \Sigma}}{C} :
    \tyrunner{\sig}{\sig'}{S}{C}
  }
\end{mathpar}
  \caption{Value typing rules.}
  \label{fig:typing-values}
\end{figure}

\begin{figure}[tp]
  \centering
  \small
\begin{mathpar}
  \coopinfer{TyUser-Return}{
    \Gamma \types V : X
  }{
    \Gamma \types \tmreturn{V} : \tyuser{X}{\Ueff}
  }

  \coopinfer{TyUser-Apply}{
    \Gamma \types V : \tyfun{X}{\tyuser{Y}{\Ueff}} \\
    \Gamma \types W : X
  }{
    \Gamma \types \tmapp{V}{W} : \tyuser{Y}{\Ueff}
  }

  \coopinfer{TyUser-Try}{
    \Gamma \types M : \tyuser{X}{(\sig,E)}
    \\
    \Gamma, x \of X \types N : \tyuser{Y}{(\sig,E')}
    \\
    \big(
      \Gamma \types N_e : \tyuser{Y}{(\sig,E')}
    \big)_{e \in E}
  }{
    \Gamma \types
    \tmtry{M}{
        \{ \tmreturn{x} \mapsto N,
           (\tmraise{e} \mapsto N_e)_{e \in E} \}
        }
    : \tyuser{Y}{(\sig,E')}
  }

  \coopinfer{TyUser-MatchPair}{
    \Gamma \types V : \typrod{X}{Y} \\
    \Gamma, x \of X, y \of Y \types M : \tyuser{Z}{\Ueff}
  }{
    \Gamma \types \tmmatch{V}{\tmpair{x}{y} \mapsto M} : \tyuser{Z}{\Ueff}
  }

  \coopinfer{TyUser-MatchEmpty}{
    \Gamma \types V : \tyempty
  }{
    \Gamma \types \tmmatch[Z]{V}{} : \tyuser{Z}{\Ueff}
  }

  \coopinfer{TyUser-MatchSum}{
    \Gamma \types V : X + Y \\
    \Gamma, x \of X \types M : \tyuser{Z}{\Ueff} \\
    \Gamma, y \of Y \types N : \tyuser{Z}{\Ueff} \\
  }{
    \Gamma \types \tmmatch{V}{\tminl{x} \mapsto M, \tminr{y} \mapsto N} : \tyuser{Z}{\Ueff}
  }

  \coopinfer{TyUser-Op}{
    \Ueff \equiv (\sig,E) \\
    \op \in \sig \\
    \Gamma \types V : A_\op \\\\
    \Gamma, x \of B_\op \types M : \tyuser{X}{\Ueff} \\
    \big(
      \Gamma \vdash N_e : \tyuser{X}{\Ueff}
    \big)_{e \in E_\op}
  }{
    \Gamma \types \tmop{op}{X}{V}{\tmcont x M}{\tmexccont N e {E_\op}} : \tyuser{X}{\Ueff}
  }

  \coopinfer{TyUser-Raise}{
     e \in E
  }{
     \Gamma \types \tmraise[X]{e} : \tyuser{X}{(\sig, E)}
  }

  \coopinfer{TyUser-Run}{
    F \equiv
    \{ \tmreturn{x} \at c \mapsto N,
       (\tmraise{e} \at c \mapsto N_e)_{e \in E},
       (\tmkill{s} \mapsto N_s)_{s \in S}
    \}
    \\\\
    \Gamma \types V : \tyrunner{\sig}{\sig'}{S}{C} \\
    \Gamma \types W : C \\\\
    \Gamma \types M : \tyuser{X}{(\sig, E)} \\
    \Gamma, x \of X, c \of C \types N : \tyuser{Y}{(\sig', E')} \\
    \big(
       \Gamma, c \of C \types N_e : \tyuser{Y}{(\sig', E')}
    \big)_{e \in E} \\
    \big(
       \Gamma \types N_s : \tyuser{Y}{(\sig', E')}
    \big)_{s \in S} \\
  }{
    \Gamma \types \tmrun{V}{W}{M}{F} : \tyuser{Y}{(\sig', E')}
  }

  \coopinfer{TyUser-Kernel}{
    F \equiv
    \{ \tmreturn{x} \at c \mapsto N,
       (\tmraise{e} \at c \mapsto N_e)_{e \in E},
       (\tmkill{s} \mapsto N_s)_{s \in S}
    \}
    \\\\
    \Gamma \types K : \tykernel{X}{(\sig, E, S, C)} \\
    \Gamma \types W : C \\
    \Gamma, x \of X, c \of C \types N : \tyuser{Y}{(\sig, E')} \\
    \big(
      \Gamma, c \of C \types N_e : \tyuser{Y}{(\sig, E')}
    \big)_{e \in E} \\
    \big(
      \Gamma \types N_s : \tyuser{Y}{(\sig, E')}
    \big)_{s \in S} \\
  }{
    \Gamma \types \tmkernel{K}{W}{F} : \tyuser{Y}{(\sig, E')}
  }
\end{mathpar}
  \caption{User typing rules.}
  \label{fig:typing-user}
\end{figure}

\begin{figure}[tp]
  \centering
  \small
\begin{mathpar}

  \coopinfer{TyKernel-Return}{
    \Gamma \types V : X
  }{
    \Gamma \types \tmreturn[C]{V} : \tykernel{X}{(\sig, E, S, C)}
  }

  \coopinfer{TyKernel-Apply}{
    \Gamma \types V : \tyfun{X}{\tykernel{Y}{\Keff}} \\
    \Gamma \types W : X
  }{
    \Gamma \types \tmapp{V}{W} : \tykernel{Y}{\Keff}
  }

  \coopinfer{TyKernel-Try}{
    \Gamma \types K : \tykernel{X}{(\sig, E, S, C)}
    \\
    \Gamma, x \of X \types L : \tykernel{Y}{(\sig, E', S, C)}
    \\
    \big(
      \Gamma \types L_e : \tykernel{Y}{(\sig, E', S, C)}
    \big)_{e \in E}
  }{
    \Gamma \types
    \tmtry{K}{
        \{ \tmreturn{x} \mapsto L,
           (\tmraise{e} \mapsto L_e)_{e \in E} \}
        }
    : \tykernel{Y}{(\sig, E', S, C)}
  }

  \coopinfer{TyKernel-MatchPair}{
    \Gamma \types V : \typrod{X}{Y} \\
    \Gamma, x \of X, y \of Y \types K : \tykernel{Z}{\Keff}
  }{
    \Gamma \types \tmmatch{V}{\tmpair{x}{y} \mapsto K} : \tykernel{Z}{\Keff}
  }

  \coopinfer{TyKernel-MatchEmpty}{
    \Gamma \types V : \tyempty
  }{
    \Gamma \types \tmmatch[Z \at C]{V}{} : \tykernel{Z}{(\sig, E, S, C)}
  }

  \coopinfer{TyKernel-MatchSum}{
    \Gamma \types V : X + Y \\
    \Gamma, x \of X \types K : \tykernel{Z}{\Keff} \\
    \Gamma, y \of Y \types L : \tykernel{Z}{\Keff} \\
  }{
    \Gamma \types \tmmatch{V}{\tminl{x} \mapsto K, \tminr{y} \mapsto L} : \tykernel{Z}{\Keff}
  }

  \coopinfer{TyKernel-Op}{
    \Keff \equiv (\sig, E, S, C) \\
    \op \in \sig \\
    \Gamma \types V : A_\op \\\\
    \Gamma, x \of B_\op \types K : \tykernel{X}{\Keff} \\
    \big(
      \Gamma \vdash L_e : \tykernel{X}{\Keff}
    \big)_{e \in E_\op}
  }{
    \Gamma \types \tmop{op}{X}{V}{\tmcont x K}{\tmexccont L e {E_\op}} : \tykernel{X}{\Keff}
  }

  \coopinfer{TyKernel-Raise}{
     e \in E
  }{
     \Gamma \types \tmraise[X \at C]{e} : \tykernel{X}{(\sig, E, S, C)}
  }

  \coopinfer{TyKernel-Kill}{
     s \in S
  }{
     \Gamma \types \tmkill[X \at C]{s} : \tykernel{X}{(\sig, E, S, C)}
  }

  \coopinfer{TyKernel-Getenv}{
    \Gamma, c \of C \types K : \tykernel{X}{(\sig, E, S, C)}
  }{
    \Gamma \types \tmgetenv[C]{\tmcont c K} : \tykernel{X}{(\sig, E, S, C)}
  }

  \coopinfer{TyKernel-Setenv}{
    \Gamma \types V : C \\
    \Gamma \types K : \tykernel{X}{(\sig, E, S, C)}
  }{
    \Gamma \types \tmsetenv{V}{K} : \tykernel{X}{(\sig, E, S, C)}
  }

  \coopinfer{TyKernel-User}{
   \Keff \equiv (\sig, E', S, C) \\
   \Gamma \types M : \tyuser{X}{(\sig, E)} \\\\
   \Gamma, x \of X \types K : \tykernel{Y}{\Keff} \\
   \big(
     \Gamma \types L_e : \tykernel{Y}{\Keff}
   \big)_{e \in E}
  }{
    \Gamma \types
    \tmuser{M}{
      \{ \tmreturn{x} \mapsto K,
         (\tmraise{e} \mapsto L_e)_{e \in E}
      \}
    }
    : \tykernel{Y}{\Keff}
  }
\end{mathpar}
  \caption{Kernel typing rules.}
  \label{fig:typing-kernel}
\end{figure}



\section{Equational theory of $\lambdacoop$}
\label{sec:appendix-equational-rules}

Values, user and kernel computations each have an equality judgement
\begin{equation*}
\Gamma \types V \equiv W : X
\qquad
\Gamma \types M \equiv N : \tyuser{X}{\Ueff}
\qquad
\Gamma \types K \equiv L  : \tyuser{X}{\Keff}.
\end{equation*}
It is presupposed that we only compare well-typed expressions with the indicated types.
For the most part, the context and the type annotation will play no part in the equation,
and so we shall drop them when no confusion can arise.

The \emph{computational equations} are displayed in \cref{fig:computational-equations-user,fig:computational-equations-kernel}.
These can be read left-to-right as evaluation rules that explain the operational meaning
of computations. The remaining equations are displayed in \cref{fig:other-equations}.
We omit standard equations which specify how substitution is performed, as well
as equations stating that equality is a congruence with respect to all the term
formers.

\begin{figure}[tbp]
  \centering
  \parbox{\textwidth}{
  \small
  \mathtoolsset{original-shortintertext=false,below-shortintertext-sep=0pt,above-shortintertext-sep=0pt}
  \begin{align*}
    \tmapp{(\tmfun{x \of X}{M})}{V} &\equiv M[V/x]
    \\[1ex]
    \tmtry{(\tmreturn{V})}{H} &\equiv N[V/x]
    \\
    \tmtry{(\tmraise[X]{e})}{H} &\equiv N_{e}
    \\
    \tmtry{(
      \tmop{op}{X}{V}{\tmcont x M}{\tmexccont {N'} {e'} {E_\op}}
    )}{H} &\equiv \\
    \tmop{op}{X}{V}{\tmcont x {\tmtry{&M}{H}}}{\left(\tmtry{N'_{e'}}{H}\right)_{e' \in E_\op}}
    \\[1ex]
    \tmmatch{\tmpair{V}{W}}{\tmpair{x}{y} \mapsto M} &\equiv
    M[V/x, W/y]
    \\
    \tmmatch[X]{V}{} &\equiv
    N
    \\
    \mathllap{
      \tmmatch{(\tminl[X,Y]{V})}{\tminl{x} \mapsto M, \tminr{y} \mapsto N} 
    } &\equiv
    M[V/x]
    \\
    \mathllap{
      \tmmatch{(\tminr[X,Y]{W})}{\tminl{x} \mapsto M, \tminr{y} \mapsto N}
    } &\equiv
    N[W/y]
    \\[1ex]
    \tmrun{V}{W}{(\tmreturn{V'})}{F} &\equiv N[V'/x, W/c]
    \\
    \tmrun{V}{W}{(\tmraise[X]{e})}{F} &\equiv N_{e}[W/c]
    \\
    \omit\rlap{$
    \tmkw{using}\; R \at W \;\tmkw{run}\;
      \tmop{op}{X}{V}{\tmcont x M}{\tmexccont {N'} {e'} {E_\op}}
      \;\tmkw{finally}\; F \equiv $}
    \\
    &
    \quad\tmkernel{K_\op[V/x]}{W}{F'}
    \\
    \shortintertext{\hfil \parbox{0.6\textwidth}{
    \abovedisplayskip=0pt
    \belowdisplayskip=0pt
    \begin{equation*}
      \text{where\ } F' \defeq 
      \begin{aligned}[t]
        \{
        &\tmreturn{x} \at c' \mapsto (\tmrun{R}{c'}{M}{F}), \\
        &\left(
            \tmraise{e'} \at c' \mapsto (\tmrun{R}{c'}{N'_{e'}}{F})
         \right)_{e' \in E_\op},\\
        &\left(
            \tmkill{s} \mapsto N_s
          \right)_{s \in S}
      \}
      \end{aligned}
    \end{equation*}
    }}
  \\
    \tmkernel{(\tmreturn[C]{V})}{W}{F} &\equiv N[V/x, W/c] \\
    \tmkernel{(\tmraise[X \at C]{e})}{W}{F} &\equiv N_{e}[W/c] \\
    \tmkernel{(\tmkill[X \at C]{s})}{W}{F} &\equiv N_{s} \\
    \tmkernel{(\tmgetenv[C]{\tmcont c K})}{W}{F} &\equiv \tmkernel{K[W/c]}{W}{F} \\
    \tmkernel{(\tmsetenv{V}{K})}{W}{F} &\equiv \tmkernel{K}{V}{F}
    \\
    \omit\rlap{$
       \tmkernel{\tmop{op}{X}{V}{\tmcont x K}{\tmexccont L {e'} {E_\op}}}{W}{F}
      \equiv $}
    \\
      \tmop{op}{X}{V}{\tmcont x {\tmkernel{K}{W}{&F}}}{
      \left(
         \tmkernel{L_{e'}}{W}{F}
      \right)_{e' \in E_\op}}
  \end{align*}
  Abbreviations:
  \begin{align*}
    F &\defeq
       \{ \tmreturn{x} \at c \mapsto N,
       (\tmraise{e} \at c \mapsto N_e)_{e \in E},
       (\tmkill{s} \mapsto N_s)_{s \in S}
       \}
    \\
    H &\defeq
       \{ \tmreturn{x} \mapsto N,
          (\tmraise{e} \mapsto N_e)_{e \in E}
       \}
    \\
    R &\defeq \tmrunner{(\tm{op}\,x \mapsto K_{\tm{op}})_{\tm{op} \in \sig}}{C}
  \end{align*}
  } 
  \caption{Computational equations (user mode).}
  \label{fig:computational-equations-user}
\end{figure}

\begin{figure}[tb]
  \centering
  \parbox{\textwidth}{
  \small
  \begin{align*}
    \tmapp{(\tmfunK{x \of X}{K})}{V} &\equiv K[V/x]
    \\[1ex]
    \tmtry{(\tmreturn{V})}{G} &\equiv L[V/x]
    \\
    \tmtry{(\tmraise[X \at C]{e})}{G} &\equiv L_{e}
    \\
    \tmtry{(\tmkill[X \at C]{s})}{G} &\equiv \tmkill[X \at C]{s}
    \\
    \tmtry{(
      \tmop{op}{X}{V}{\tmcont x K}{\tmexccont {L'} {e'} {E_\op}}
    )}{G} &\equiv \\
    \tmop{op}{X}{V}{\tmcont x {\tmtry{&K}{G}}}{\left(\tmtry{L'_{e'}}{G}\right)_{e' \in E_\op}}
    \\
    \tmtry{(\tmgetenv[C]{\tmcont c K})}{G} &\equiv \tmgetenv[C]{\tmcont c {\tmtry{K}{G}}}
    \\
    \tmtry{(\tmsetenv{V}{K})}{G} &\equiv \tmsetenv{V}{\tmtry{K}{G}}
    \\[1ex]
    \tmmatch{\tmpair{V}{W}}{\tmpair{x}{y} \mapsto K} &\equiv
    K[V/x, W/y]
    \\
    \tmmatch[X \at C]{V}{} &\equiv
    K
    \\
    \tmmatch{(\tminl[X,Y]{V})}{\tminl{x} \mapsto K, \tminr{y} \mapsto L} &\equiv
    K[V/x]
    \\
    \tmmatch{(\tminr[X,Y]{W})}{\tminl{x} \mapsto K, \tminr{y} \mapsto L} &\equiv
    L[W/y]
    \\[1ex]
    \tmuser{(\tmreturn{V})}{G} &\equiv L[V/x]
    \\
    \tmuser{(\tmraise[X]{e})}{G} &\equiv L_e
    \\
    \tmuser{(
      \tmop{op}{X}{V}{\tmcont x M}{\tmexccont {N'} {e'} {E_\op}}
    )}{G} &\equiv \\
    \tmop{op}{X}{V}{\tmcont x {\tmuser{&M}{G}}}{\left(\tmuser{N'_{e'}}{G} \right)_{e' \in E_\op}}
  \end{align*}
  Abbreviation: $G \defeq \{ \tmreturn{x} \mapsto L, (\tmraise{e} \mapsto L_e)_{e \in E} \}$
  } 
  \caption{Computational equations (kernel mode).}
  \label{fig:computational-equations-kernel}
\end{figure}

\begin{figure}[tb]
  \centering
  \parbox{\textwidth}{
  \small
  \begin{gather*}
    V \equiv \tmunit : \tyunit \qquad\qquad
    \tmfun{x \of A}{\tmapp{V}{x}} \equiv V \qquad\qquad
    \tmfunK{x \of A}{\tmapp{V}{x}} \equiv V
    \\
    \tmtry{M}{
       \{ \tmreturn{x} \mapsto \tmreturn{x},
          (\tmraise{e} \mapsto \tmraise[X]{e})_{e \in E}
       \}
    }
    \equiv M
    \\
    \tmtry{K}{
       \{ \tmreturn{x} \mapsto \tmreturn{x},
          (\tmraise{e} \mapsto \tmraise[X \at C]{e})_{e \in E}
       \}
    }
    \equiv K
    \\[1ex]
  \begin{aligned}
  \tmgetenv[C]{\tmcont c {\tmsetenv{c}{K}}} &\equiv K
  \\
  \tmsetenv{V}{\tmgetenv[C]{\tmcont c K}} &\equiv
  \tmsetenv{V}{K[V/c]}
  \\
  \tmsetenv{V}{\tmsetenv{W}{K}} &\equiv \tmsetenv{W}{K}
  \\
  \tmgetenv[C]{\tmcont c {\tmkill[X \at C]{s}}} &\equiv
  \tmkill[X \at C]{s}
  \\
  \tmsetenv{V}{\tmkill[X \at C]{s}} &\equiv
  \tmkill[X \at C]{s}
  \\
  \tmgetenv[C]{\tmcont c {\tmop{op}{X}{V}{\tmcont x K}{\tmexccont L e {E_\op}}}}
  &\equiv \\
  \tmop{op}{X}{V}{\tmcont x {\tmgetenv[C]{&\tmcont c K}}}{\left(\tmgetenv[C]{\tmcont c {L_e}}\right)_{e \in E_\op}}
  \\
  \tmsetenv{V}{\tmop{op}{X}{V}{\tmcont x K}{\tmexccont L e {E_\op}}}
  &\equiv \\
  \tmop{op}{X}{&V}{\tmcont x {\tmsetenv{V}{K}}}{\left(\tmsetenv{V}{L_e}\right)_{e \in E_\op}}
  \end{aligned}
  \end{gather*}
  } 
  \caption{Other equations (for $\eta$-expansion and the kernel theory from \cref{sec:user-kernel-monads}).}
  \label{fig:other-equations}
\end{figure}


\end{document}